\documentclass[onecolumn,12pt]{IEEEtran}
\usepackage{setspace}  
\usepackage{fullpage}
\usepackage{ifpdf}
\usepackage{graphicx}
\usepackage{caption}
\usepackage{subcaption}
\usepackage{amsmath}
\usepackage{amsthm}
\usepackage{amsfonts}
\usepackage{amssymb}
\usepackage{eucal}
\usepackage{array}
\usepackage{url}
\usepackage{bm}
\usepackage{csquotes}
\usepackage{multirow}
\usepackage{dblfloatfix}
\usepackage[colorlinks=true,pdfstartview=FitV,linkcolor=black,citecolor=black,urlcolor=blue,plainpages=false]{hyperref}
\usepackage[numbers,sort&compress]{natbib}
\usepackage{enumerate}
\theoremstyle{remark}
\newtheorem{theorem}{Theorem}

\newtheorem{remark}{Remark}
\newtheorem{Proposition}{Proposition}
\newtheorem{corollary}{Corollary}
\newcommand\norm[1]{\left\lVert#1\right\rVert}
\usepackage{mathtools}
\DeclarePairedDelimiter{\abs}{\lvert}{\rvert}

\begin{document}
 \author{
      Jingwen Bai and Ashutosh Sabharwal, \emph{Fellow, IEEE}\footnote{J. Bai and A. Sabharwal are with Department of Electrical and Computer Engineering
      Rice University, Houston, TX 77005, USA, e-mail:\{jingwen, ashu\}@rice.edu. This work was partially supported by NSF CNS-1012921, NSF CNS-1161596 and Intel.}}


\title{Large Antenna Analysis of Multi-Cell Full-Duplex Networks}

\vspace{15pt}

\maketitle
\thispagestyle{plain}
\pagestyle{plain}

\begin{abstract}
We study a multi-cell multi-user MIMO full-duplex network, where each base station~(BS) has multiple antennas with full-duplex capability supporting single-antenna users with either full-duplex or half-duplex radios. We characterize the up- and downlink ergodic achievable rates for the case of linear precoders and receivers. The rate analysis includes practical constraints such as imperfect self-interference cancellation, channel estimation error, training overhead and pilot contamination. We show that the 2$\times$ gain of full-duplex over half-duplex system remains in the asymptotic regime where the number of BS antennas grows infinitely large. We numerically evaluate the finite SNR and antenna performance, which reveals that full-duplex networks can use significantly fewer antennas to achieve spectral efficiency gain over the half-duplex counterparts. 
In addition, the overall full-duplex gains can be achieved under realistic 3GPP multi-cell network settings despite the increased interference introduced in the full-duplex networks.

\end{abstract}
\section{Introduction}
One of the emerging techniques to significantly improve the spectral efficiency in wireless networks is full-duplex wireless communication~\cite{ashu13}.
In-band full-duplex wireless allows simultaneous transmission and reception using the same frequency band, and thus opens up new design opportunities to increase the spectral efficiency of wireless systems. The feasibility of a (near-)full-duplex radio has been demonstrated by many groups, see e.g.\cite{ashu13,duarte2012design,evan13,Bharadia2013,aryafar2012midu,Kim15,HossainH15} and references therein. 
A side-effect of the full-duplex operation is that additional interference is introduced because there are more simultaneous active links and hence there is a possibility that the full-duplex gain can be offset by the loss due to additional interference. In the example shown in 
Fig.~\ref{fig:topology}, the uplink rate will be affected by the new interference from neighboring full-duplex BSs, and downlink rate will be affected by the new interference from uplink users (UE).


\begin{figure}
\centering
   \includegraphics[width=0.65\textwidth,trim = 20mm
      45mm 20mm 23mm, clip]{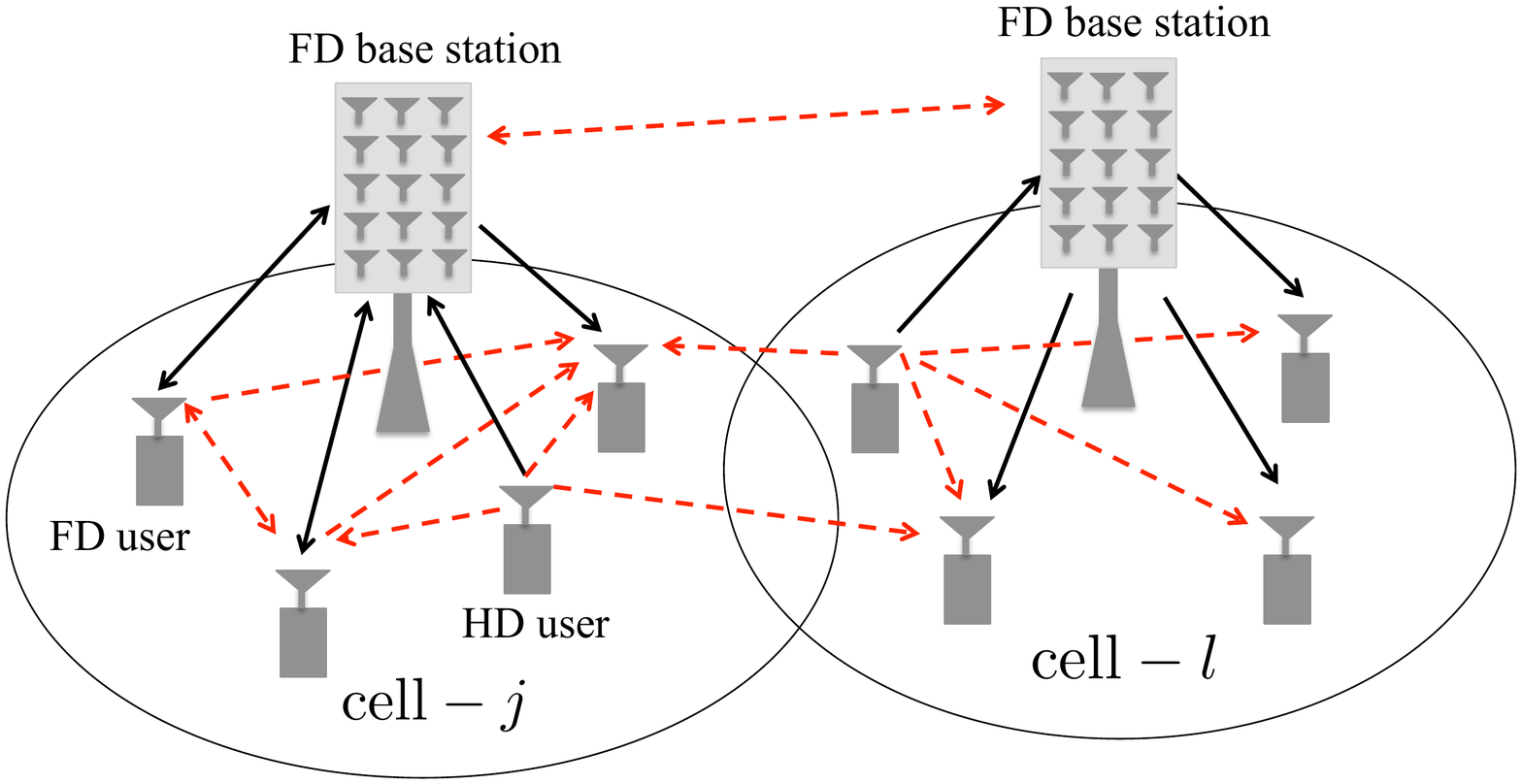}
\caption{The full-duplex BS in each cell has $M$ antennas, and the UE has single antenna with either full-duplex or half-duplex radio. Besides the conventional interference, there will be new BS-BS and UE-UE interference highlighted by the red dash lines.}
\label{fig:topology}
\end{figure}

In this paper, we study if and how large antenna arrays at BSs can be used to manage the increased intra- and inter-cell interference in full-duplex enabled networks. Recently, the use of a very large antenna array at the BS has become very attractive~\cite{Shepard2013,larsson2014massive,rusek2013scaling}, known as massive MIMO,  where a BS has orders of magnitude more antennas compared with the current use. The large antenna array at the BS not only can increase the network capacity many-fold, but also enable a new network architecture to simplify baseband signal processing~\cite{Shepard2013}, eliminate inter-cell interference~\cite{marzettaNon10}, and reduce node transmit power for energy saving~\cite{Ngo13}.  The experimental evidence on the benefits of massive MIMO
 has already sparked strong industry interest and 64-antenna configuration~\cite{3gppMassiveMIMO} is now being considered for 5G systems. 


Our contributions are three-fold. First, we provide a general analysis to characterize the uplink and downlink ergodic achievable rates in multi-cell multi-user MIMO~(MU-MIMO) full-duplex networks. Focusing on computationally efficient linear receivers and precoders, we consider the case where each BS has multiple antennas with full-duplex capability, while each UE has a single antenna with either full-duplex or half-duplex radio. Practical constraints such as imperfect self-interference cancellation, channel estimation error, training overhead and pilot contamination are considered in our analysis. 

Second, we analyze the system performance in the asymptotic regime where the number of BS antennas grows infinitely large. We show that the transmit power of BSs and UEs can be scaled down with an increasing number of BS antennas to maintain a fixed asymptotic rate. The impact of imperfect self-interference cancellation at full-duplex BSs and full-duplex UEs, intra-cell and inter-cell interference in the multi-cell MU-MIMO full-duplex networks disappears as the number of BS antennas becomes infinitely large. Under the assumption of perfect channel knowledge, full-duplex system asymptotically achieves 2$\times$ spectral efficiency gain over the half-duplex system. When channel estimation error and channel training overhead are considered, the 2$\times$ asymptotic full-duplex gain is achieved when serving only full-duplex UEs. 

Lastly, we numerically evaluate the system performance in finite SNR and finite antenna regimes. Our numerical results reveal that full-duplex networks can use significantly fewer antennas to achieve spectral efficiency gain over the half-duplex counterparts. 
In addition, under realistic 3GPP multi-cell network scenarios~\cite{3gppdtdd}, 
the overall full-duplex gains can be achieved despite the increased interference introduced in the full-duplex networks.

\textbf{Related work:} 
There are two lines of work in the area of MU-MIMO full-duplex cellular networks. The first line of work focuses on information theoretic limits while the other line of work focuses on practical network design. Towards that end, we discuss \cite{Sahai13, Bai15, Jeon15, Karakus15, Wenzhuo15, JingwenTWC, Goyal13,Choi13,Amir15}; note that not all papers are multi-cell studies. In \cite{Sahai13, Bai15, Jeon15,Karakus15,Wenzhuo15}, only a single-cell MIMO full-duplex case is studied. 
The authors in \cite{Sahai13, Bai15, Jeon15} study the information theoretic limits of a single-cell MU-MIMO full-duplex network, where the high signal-to-noise~(SNR) approximation, i.e., degrees of freedom or multiplexing gains have been characterized. In \cite{Karakus15,Wenzhuo15}, the authors propose opportunistic scheduling and a distributed power control method to mitigate UE-UE interference in a single-cell MIMO full-duplex network. Authors in \cite{JingwenTWC} present several schemes to manage UE-UE interference via out-of-band wireless  side-channels in the full-duplex network. 

All the papers~\cite{Goyal13,Choi13,Amir15} study multi-cell full-duplex networks. In \cite{Goyal13}, scheduling and power control algorithms with BS cooperation in the multi-cell SISO full-duplex networks are investigated and \cite{Choi13} studies the MIMO case with 4 BS antennas. The performance analyses are based on extensive simulations, largely because of the challenge of analyzing complex scheduling and power control methods. In~\cite{Amir15}, the authors study the degrees of freedom region when the BSs have full coordination. This converts the problem into a network MIMO problem, and essentially allows one to treat the multi-cell problem as one giant MIMO cell. This, in turn, allows the use of interference alignment to achieve the highest possible degrees of freedom. While this approach provides insights into the maximum possible degrees of freedom, it relies on full coordination and proposes a very complex transmission method - both of them are extremely challenging to implement in practice.  Compared to the above works, we focus on the case when there is no BS coordination and hence we cannot convert the problem to the more tractable single-cell problem. In addition, we allow only simple linear processing at the BSs, namely conjugate beamforming, and hence complex schemes like zero-forcing or interference-alignment cannot be employed. 


The rest of paper is organized as follows. We first describe the system model in Section~\ref{system}. In Section~\ref{core}, we characterize the ergodic achievable rates of uplink and downlink under both perfect and imperfect CSI assumptions. The large-scale system performance is studied in Section~\ref{asym}. The numerical results with realistic network evaluation are presented in Section~\ref{numerical}. Section~\ref{conclusion} concludes the paper.

\section{Multi-cell Full-Duplex System Model}\label{system}

We consider a  multi-cell MU-MIMO full-duplex system with $L~(\geq 1)$ cells, where each cell has one in-band full-duplex BS with $M$ antennas. In each cell, single-antenna UEs with either full-duplex or half-duplex radios are supported. Each BS can serve $K_f$ full-duplex~(FD) users, $K_{h}^{u}$ half-duplex~(HD) uplink users and $K_{h}^d$ half-duplex downlink users. We denote $\mathcal{K}_{u}=\{\underbrace{1,\cdots,K_{f}}_{\mathrm{FD~UE}},\underbrace{K_{f}+1,\cdots,K_f+K_h^u}_{\mathrm{HD~UE}}\}$ as the set of all uplink users,  where the first $K_f$ elements represent $K_f$ full-duplex users; and the last $K_{h}^{u}$ elements represent half-duplex uplink users. The set of all downlink users are denoted as $\mathcal{K}_d=\{\underbrace{1,\cdots,K_{f}}_{\mathrm{FD~UE}},\underbrace{K_{f}+1,\cdots,K_f+K_h^d}_{\mathrm{HD~UE}}\}$, where the first $K_f$ elements represent full-duplex users, and  the last $K_{h}^{d}$ elements represent half-duplex downlink users. We further denote the set of full-duplex users as $\mathcal{K}_f\triangleq \{1,\cdots,K_f\}$ and set of half-duplex downlink users as $\mathcal{K}_h^d\triangleq \{K_f+1,\cdots,K_f+K_h^d\}$. The total number of uplink users is $|\mathcal{K}_{u}|\triangleq K_u$, where $K_u=K_{f}+K_h^u$ and the total number of downlink users is $|\mathcal{K}_{d}|\triangleq K_d$, where $K_d=K_{f}+K_h^d$. 

The uplink and downlink data are transmitted over the same time-frequency slot with block fading. The analysis in this paper can be applied to wide-band channels like OFDM system.

In this work, we will consider practical constraints on the full-duplex radio chains such as non-ideal power amplifier, oscillator phase noise, non-ideal digital-to-analog converter and analog-to-digital converter, which can be captured by the dynamic range model~\cite{Day12,Vehkapera13}. The dynamic range model approximates the imperfect full-duplex transmit radio chain as an additive white Gaussian ``transmitter noise" added to each transmit antenna. The variance of the transmitter noise is $\kappa$~($\kappa\ll 1$) times the power of the transmit signals, where $\kappa$ is the dynamic range parameter.
The full-duplex  transmitter noise will propagate over the self-interference~(SI) channel and become nontrivial compared to the receiver thermal noise. However, the effect of transmitter noise that propagates over the uplink/downlink channels can be ignored compared to the receiver thermal noise~\cite{Vehkapera13}.
\subsection{Uplink}
The $j$-th full-duplex BS will receive an $M\times 1$ uplink signal vector $\bm{y}_{u,j}^\prime$:
\begin{gather}
\begin{aligned}
\bm{y}_{u,j}^\prime=\sum_{l=1}^L\bm{G}_{u,jl} \bm{x}_{u,l}+\sum_{l=1}^{L}\bm{V}_{jl}\bm{x}_{d,l}+\bm{V}_{jj}\bm{e}_{bs,j}+\bm{n}_{u,j},
\label{original:uplink}
\end{aligned}
\end{gather}
where $\bm{x}_{u,l} = \sqrt{P_u}\bm{u}_l$ denotes the uplink transmit signal vector, $\bm{u}_l\triangleq[u_{l,1},\cdots,u_{l,K_u}]^T$ is a $K_u\times 1$ vector consisting of uplink messages of the $K_u$ uplink users in the $l$-th cell. Each user has an average power constraint $P_u$, and $\mathbb{E}(|{u}_{l,k}|^2)=1$, for $k\in\mathcal{K}_u$.
$\bm{x}_{d,l}$ is an $M\times 1$ vector denoting downlink transmit signal in the $l$-th cell. Each BS has an average power constraint $P_d$.

We assume $\bm{G}_{u,jl}$ is an $M\times K_u$ matrix denoting the channel between the uplink users in the $l$-th cell and the $j$-th BS.
The propagation channel model in our system considers both small-scale fading due to mobility and multipath, and large-scale fading due to geometric attenuation and shadowing effect, thus allowing arbitrary cell layout.

The uplink channel $\bm{G}_{u,jl}$ encompasses independent small-scale fading and large-scale fading 
\[(\bm{G}_{u,jl})_{m,n}\triangleq g_{u,jmln}=h_{u,jmln}\sqrt{\beta_{u,jln}},\]
where $j,l\in\{1,\cdots,L\},m\in\{1,\cdots,M\},n\in\mathcal{K}_{u}$. $h_{u,jmln}$ is the small-scale fading value between the $n$-th uplink user in the $l$-th cell and the $m$-th antenna at the $j$-th BS, following independent and identically distributed~(i.i.d.) circularly-symmetric complex Gaussian distribution $\mathcal{CN}(0,1)$. We use ${\beta_{u,jln}}$ to model the path loss and shadow fading between the $n$-th uplink user in the $l$-th cell and the $j$-th BS, which is independent of $m$. Such long-term parameters can be measured at the BS. 

The BS-BS channel $\bm{V}_{j,l}$ is an $M\times M$ matrix denoting the channel between the $l$-th BS and the $j$-th BS, and $\bm{V}_{j,j}$ denotes the SI channel for the $j$-th full-duplex BS. We assume $\bm{V}_{j,l}$ contains i.i.d $\mathcal{CN}(0,\beta_{b,jl})$ elements.
The imperfect full-duplex transmit front-end chain is modeled as transmit noise $\bm{e}_{bs,j}\in\mathbb{C}^{M\times 1}$ added to each transmit antenna at the $j$-th full-duplex BS. And $\bm{e}_{bs,j}$ will propagate through the SI channel and follow~$\mathcal{CN}\left(0,\frac{\kappa P_d}{M}\bm{I}_M\right)$ assuming equal power allocation among downlink users, where $\bm{I}_M$ denotes an $M\times M$ identity matrix.
The receiver thermal noise is denoted as $\bm{n}_{u,j}$ which contains i.i.d $\mathcal{CN}(0,\sigma^2)$ entries. 

The $j$-th BS then performs SI cancellation by knowing its SI channel and its downlink signal. Hence from (\ref{original:uplink}), we have 
\begin{gather}
\begin{aligned}
\bm{y}_{u,j}=\sum_{l=1}^L\bm{G}_{u,jl} \bm{x}_{u,l}+\sum_{l\neq j}^{L}\bm{V}_{j,l}\bm{x}_{d,l}+\bm{z}_{u,j},\label{eq:ul}
\end{aligned}
\end{gather}
where $\bm{z}_{u,j}\sim \mathcal{CN}\left(0,\left(\sigma^2+\kappa P_d\beta_{b,jj}\right)\bm{I}_M\right)$.

\subsection{Downlink}
The received signals for the $K_d$ downlink users in the $l$-th cell can be expressed as a $K_d\times 1$ vector $\bm{y}_{d,l}$, given by
\begin{gather}
\begin{aligned}
\bm{y}_{d,l}=\sum_{j=1}^L\bm{G}^T_{d,jl}\bm{x}_{d,j}+\sum_{j=1}^L\bm{F}_{lj}\bm{x}_{u,j}+\bm{J}_{l}+\bm{n}_{d,l}, \label{eq:dl}
\end{aligned}
\end{gather}
where $\bm{G}_{d,jl}$ is an $M\times K_d$ matrix denoting the channel between the downlink users in the $l$-th cell and the $j$-th BS, the downlink channel can be represented as $\bm{G}_{d,jl}^T$ and we have
\[(\bm{G}_{d,jl})_{m,k}^T\triangleq g^T_{d,jmlk}=h^T_{d,jmlk}\sqrt{\beta_{d,jlk}},\]
 where $k\in\mathcal{K}_{d}$, the superscript $``T"$ denotes transpose. $h_{d,jmlk}$ is the small-scale fading value between the $k$-th downlink user in the $l$-th cell and the $m$-th antenna at the $j$-th BS, following $\mathcal{CN}(0,1)$; ${\beta_{d,jlk}}$ incorporates path loss and shadow fading between the $k$-th downlink user in the $l$-th cell and the $j$-th BS .

The UE-UE interference channel $\bm{F}_{lj}$ is a $K_d\times K_u$ matrix denoting the interference channel from $K_u$ uplink users in the $j$-th to $K_d$ downlink users in the $l$-th cell; $\bm{J}_{l}=
\begin{pmatrix}
    \mathsf{diag}(\bm{S}_{ll})\bm{e}_{ue,l} \\
  \bm{0}
 \end{pmatrix}$ is a $K_d\times 1$ vector, where $\bm{S}_{ll}$ is a $K_f\times K_f$ matrix denoting the interference channels between full-duplex users in the $l$-th cell. The diagonal elements of $\bm{S}_{ll}$ constitute SI channels for each full-duplex user in the $l$-th cell. And $\bm{S}_{ll}$ is a sub-matrix of $\bm{F}_{ll}$, where $ \left(\mathsf{diag}(\bm{S}_{ll})\right)_k=(\bm{F}_{ll})_{k,k}\triangleq f_{lklk},~k\in\mathcal{K}_f$. We assume each entry in $(\bm{F}_{lj})_{k,n}\triangleq f_{lkjn}$ models both the large-scale and fast fading channel coefficient between the $n$-th uplink user in the $j$-th cell and the $k$-th downlink user in the $l$-th cell, and follows i.i.d. $\sim\mathcal{CN}(0,\beta_{I,lkjn})$, where $j,l\in\{1,\cdots,L\}$ and $n\in\mathcal{K}_u,k\in\mathcal{K}_d$. 
The transmit noise at each full-duplex user in the $l$-th cell $\bm{e}_{ue,l}\in\mathbb{C}^{K_f\times 1}$ will propagate through the SI channel and follow~$\mathcal{CN}\left(0,\kappa P_u\bm{I}_{K_f}\right)$.  The receiver thermal noise $\bm{n}_{d,l}$ follows~$\sim\mathcal{CN}(0,\sigma^2\bm{I}_{K_d})$.

Next, the $k$-th full-duplex user in the $l$-th cell where $k\in\mathcal{K}_f$, will perform SI cancellation by subtracting out its own interference from the received signal in (\ref{eq:dl}). Thus we have
\begin{gather}
\begin{aligned}
{r}_{d,lk}=\sum_{j=1}^L\bm{g}^T_{d,jlk}\bm{x}_{d,j}+\sum_{j=1}^L\sum_{n=1}^{K_u}\sqrt{P_u}f_{lkjn}{u}_{j,n}-\sqrt{P_u}f_{lklk}{u}_{l,k}+{z}_{d,lk}, \label{eq:dlfd}
\end{aligned}
\end{gather}
where $\bm{g}_{d,jlk}$ is the $k$-th column of matrix $\bm{G}_{d,jl}$ and ${z}_{d,lk}\sim \mathcal{CN}\left(0,\left(\sigma^2+\kappa P_u\beta_{I,lklk}\right)\right).$ 

The expressions of the received downlink signals for full-duplex users and half-duplex users will differ, because for a full-duplex user, after SI cancellation there is no self UE-UE interference but an additional transmit noise is added to the received signal, while for a half-duplex downlink user, it will suffer the interference from all uplink users. Hence we can rewrite the received downlink signal for the $k^\prime$-th half-duplex user in the $l$-th cell where $k^\prime\in\mathcal{K}_h^d$ as
\begin{gather}
\begin{aligned}
{r}_{d,lk^\prime}=\sum_{j=1}^L\bm{g}^T_{d,jlk^\prime}\bm{x}_{d,j}+\sum_{\mathclap{\substack{j=1}}}^L\sum_{\mathclap{n=1}}^{K_u} \sqrt{P_u}f_{lk^\prime jn}u_{j,n}+{n}_{d,lk^\prime}, \label{eq:dlhd}
\end{aligned}
\end{gather}
where ${n}_{d,lk^\prime}$ is the $k^\prime$-th element in $\bm{n}_{d,l}$.

\section{Achievable rates in full-duplex networks} \label{core}
In this section, we will derive the up- and downlink ergodic achievable rates in multi-cell MU-MIMO full-duplex networks. We first study the case of perfect channel state information~(CSI) where the channel information is obtained perfectly at no cost. Next, we consider the channel estimation error, where CSI is estimated from uplink pilot sequences. We assume synchronized reception from all cells, which will reflect the worst possible scenario of pilot contamination~\cite{marzettaNon10}.  
\subsection{Perfect channel state information}
\subsubsection{Uplink with maximum ratio combining}
The $j$-th BS will receive data transmitted by its associated $K_u$ uplink users, together with the interference from other cells. 
We apply a low-complexity linear receiver, i.e., maximum ratio combing for uplink signal detection.  The $j$-th BS will multiply the received signal after SI cancellation by the conjugate-transpose of its uplink channel $\bm{G}_{u,jj}^H$ to obtain a $K_u\times 1$ signal vector 
\begin{gather}
\begin{aligned}
 \bm{r}_{u,j}&=\bm{G}_{u,jj}^H\bm{y}_{u,j}\\
 &=\bm{G}_{u,jj}^H\bm{G}_{u,jj}\bm{x}_{u,j}+\sum_{\mathclap{l\neq j}}\bm{G}_{u,jj}^H\bm{G}_{u,jl}\bm{x}_{u,l}+\sum_{l\neq j}\bm{G}_{u,jj}^H\bm{V}_{j,l}\bm{x}_{d,l}+\bm{G}_{u,jj}^H\bm{z}_{u,j}, \label{eq:ul2}
\end{aligned}
\end{gather}
where superscript ``$H$" denotes conjugate-transpose, $\bm{y}_{u,j}$ is given in (\ref{eq:ul}).

\subsubsection{Downlink with conjugate beamforming}
The $l$-th BS will transmit an $M\times 1$ downlink signal vector $\bm{x}_{d,j}$ by precoding the downlink messages using a conjugate beamforming linear precoder 
\begin{gather}
\begin{aligned}
\bm{x}_{d,l}=\frac{\bm{G}^*_{d,ll}}{\sqrt{\gamma_l}}\bm{s}_{d,l},\label{precode:p}
\end{aligned}
\end{gather}
where superscript ``*" denotes conjugate. $\bm{s}_{d,l}=\sqrt{\frac{P_d}{K_d}}\bm{d}_l$, 
$\bm{d}_l\triangleq[d_{l,1},\cdots,d_{l,K_d}]^T$ is a $K_d\times 1$ vector consisting of downlink messages of the $K_d$ downlink users in the $l$-th cell with $\mathbb{E}(|d_{l,k}|^2)=1,k\in\mathcal{K}_d$; $\gamma_l$ is the $l$-th cell power normalization factor to meet the average power constraint such that $\mathbb{E}(\bm{x}_{d,l}^H\bm{x}_{d,l})=P_d$, hence $\gamma_l=\frac{\mathbb{E}(\bm{d}_l^H\bm{G}^T_{d,ll}\bm{G}^*_{d,ll}\bm{d}_l)}{K_d}=\frac{{M\sum_{k=1}^{K_d} \beta_{d,llk}}}{K_d}$.

Substituting (\ref{precode:p}) into (\ref{eq:dlfd}), we can first obtain the downlink signal at the $k$-th full-duplex user in the $l$-th cell where $k\in\mathcal{K}_f$ as
\begin{gather}
\begin{aligned}
{r}_{d,lk}=\sum_{j=1}^L\sum_{i=1}^{K_d}\sqrt{\frac{P_d}{K_d\gamma_j}}\bm{g}^T_{d,jlk}\bm{g}^*_{d,jji}{d}_{j,i}+\sum_{j=1}^L\sum_{n=1}^{K_u}\sqrt{P_u}f_{lkjn}{u}_{j,n}-\sqrt{P_u}f_{lklk}{u}_{l,k}+{z}_{d,lk}. \label{eq:dl2}
\end{aligned}
\end{gather}

Next, we substitute (\ref{precode:p}) into (\ref{eq:dlhd}) to obtain the received downlink signal at the $k^\prime$-th half-duplex user in the $l$-th cell where $k^\prime\in\mathcal{K}_h^d$ as
\begin{gather}
\begin{aligned}
{r}_{d,lk^\prime}=\sum_{j=1}^L\sum_{i=1}^{K_d}\sqrt{\frac{P_d}{K_d\gamma_j}}\bm{g}^T_{d,jlk^\prime}\bm{g}^*_{d,jji}{d}_{j,i}+\sum_{\mathclap{\substack{j=1}}}^L\sum_{\mathclap{n=1}}^{K_u} \sqrt{P_u}f_{lk^\prime jn}u_{j,n}+{n}_{d,lk^\prime}. \label{eq:dl3}
\end{aligned}
\end{gather}

\subsubsection{Ergodic achievable rates}
We will treat the interference terms in (\ref{eq:ul2}), (\ref{eq:dl2}) and (\ref{eq:dl3}) as noise in our rate analysis. By coding over infinitely large number of the realizations of the fast ading channels, we can obtain the ergodic achievable rate of the $n$-th uplink user in the $j$-th cell~(bits/s/Hz) as
\begin{gather}
\begin{aligned}
\tilde{R}^{fd,p}_{u,jn}&= \mathbb{E}\left\{\mathrm{log}_2\left(1+\frac{P_u\norm{\bm{g}_{u,jjn}}^4}{P_u\sum_{(l,m)\neq (j,n)}|\bm{g}^H_{u,jjn}\bm{g}_{u,jlm}|^2+I_{bs-bs}+\norm{\bm{g}_{u,jjn}}^2(\sigma^2+\kappa P_d\beta_{b,jj})}\right)\right\}, \label{ul:lb1}
\end{aligned}
\end{gather}
where $n\in\mathcal{K}_u$, $I_{bs-bs}=\sum_{l\neq j}\sum_{k\in\mathcal{K}_d}\frac{P_d}{K_d \gamma_l}|\bm{g}^H_{u,jjn}\bm{V}_{jl}\bm{g}^*_{d,llk}|^2$. The notation of $\sum_{(l,m)\neq (j,n)}$ denotes the summation over all tuples $(l,m)\in\{1,\cdots,L\}\times \mathcal{K}_u\backslash\{(l=j,m=n)\}$.

Similarly,  the downlink ergodic achievable rates of the $k$-th full-duplex user and the $k^\prime$-th half-duplex user in the $l$-th cell are respectively given as 
\begin{gather}
\begin{aligned}
\tilde{R}^{fd,p}_{d,lk}&= \mathbb{E}\left\{\mathrm{log}_2\left(1+\frac{\frac{P_d}{K_d\gamma_l}\norm{\bm{g}_{d,llk}}^4}{\sum_{(j,i)\neq (l,k)}\frac{P_d}{K_d\gamma_j}|\bm{g}^T_{d,jlk}\bm{g}^*_{d,jji}|^2+I_{ue-ue}(k)- P_u|f_{lklk}|^2+\sigma^2+\kappa P_u\beta_{I,lklk}}\right)\right\},\\
\tilde{R}^{fd,p}_{d,lk^\prime}&= \mathbb{E}\left\{\mathrm{log}_2\left(1+\frac{\frac{P_d}{K_d\gamma_l}\norm{\bm{g}_{d,llk^\prime}}^4}{\sum_{(j,i)\neq (l,k^\prime)}\frac{P_d}{K_d\gamma_j}|\bm{g}^T_{d,jlk^\prime}\bm{g}^*_{d,jji}|^2+I_{ue-ue}(k^\prime)+\sigma^2}\right)\right\},\label{ARFDP}
\end{aligned}
\end{gather}
 where $k\in\mathcal{K}_f,k^\prime\in\mathcal{K}_h^d$. $I_{ue-ue}(k)=\sum_{j=1}^L\sum_{n\in\mathcal{K}_u}|P_uf_{lkjn}|^2$. The notation of $\sum_{(j,i)\neq (l,k)}$ denotes the summation over all tuples $(j,i)\in\{1,\cdots,L\}\times \mathcal{K}_d\backslash\{(j=l,i=k)\}$.

\begin{Proposition}\label{prop1}
For perfect CSI, lower bounds on the ergodic achievable rates of multi-cell MU-MIMO full-duplex networks when $M\geq 3$ are given as 
\begin{gather}
\begin{aligned}
\text{Uplink user:}~R^{fd,p}_{u,jn}&= \mathrm{log}_2\bigg(1+\frac{P_u(M-1)\beta_{u,jjn}}{I_{up}+\sigma^2+\kappa P_d\beta_{b,jj}}\bigg),\\
\text{Downlink, FD user:}~R^{fd,p}_{d,lk}&= \mathrm{log}_2\Bigg(1+\frac{\eta_l P_d(M-1)(M-2)\beta^2_{d,llk}}{I_{down}(k)-P_u\beta_{I,lklk}+\sigma^2+\kappa P_u\beta_{I,lklk}}\Bigg),\\
\text{Downlink, HD user:}~R^{fd,p}_{d,lk^\prime}&= \mathrm{log}_2\Bigg(1+\frac{\eta_l P_d(M-1)(M-2)\beta^2_{d,llk^\prime}}{I_{down}(k^\prime)+\sigma^2}\Bigg),\label{eq:prop1}
\end{aligned}
\end{gather}
where $n\in\mathcal{K}_u,~k\in\mathcal{K}_f,~k^\prime\in\mathcal{K}_h^d$, $I_{up}=P_u\sum_{(l,m)\neq (j,n)}\beta_{u,jlm}+P_d\sum_{l\neq j} \beta_{b,jl}$, $I_{down}(k) = \sum_{i\neq k,i\in\mathcal{K}_d}\eta_l P_d\beta_{d,llk}\beta_{d,lli}(M-2)+P_d\sum_{j\neq l}\beta_{d,jlk}+\sum_{j=1}^L\sum_{n\in\mathcal{K}_u}P_u\beta_{I,lkjn}$, and $\eta_l=\frac{1}{M\sum_{i\in\mathcal{K}_d}\beta_{d,lli}}$. 
\end{Proposition}
\begin{proof}
See Appendix~\ref{proof1}.
\end{proof}

\subsection{Imperfect channel state information}
In order to perform up- and downlink beamforming in MU-MIMO networks, the BS needs to know the up- and downlink channels
for uplink coherent detection and downlink precoding. In this section, we will derive the ergodic achievable rates in the presence of channel estimation error. In the full-duplex system, the up- and downlink channels are estimated through uplink training sequences, and thus the pilot overhead is only proportional to the number of users. Simultaneous up-and downlink data transmission starts after uplink training, as shown in Fig.~\ref{training}.
\subsubsection{Pilot-aided channel estimation}
During the uplink training period,  within a coherence interval of $T$, $K_{tot}\triangleq K_u+K_d-K_f$ mutually orthogonal pilot sequences of length $\tau$~($\tau\geq K_{tot}$) symbols are used to estimate the channels between each BS and its associated UEs. The same set of pilot sequences will be reused by $L$ cells. 
The channel estimate will be corrupted by pilot contamination~\cite{marzettaNon10} due to the non-orthogonality of the reused pilots.
We assume that each user has an average channel training power of $P_{tr}$, which is a parameter that depends on the length of the pilot sequences.

The $j$-th BS will correlate the received signal from uplink training with the pilot sequences assigned for the $k$-th user to obtain an $M$-dimensional vector $\bm{y}_{tr,jk}$
\begin{gather}
\bm{y}_{tr,jk}=\bm{g}_{\Phi,jjk}+\sum_{l\neq j}^L\bm{g}_{\Phi,jlk}+\frac{\bm{n}_{jk}}{\sqrt{P_{tr}}}, 
\end{gather}
where $\Phi\in\{u,d\}$, $k\in\mathcal{K}_u\cup\mathcal{K}_h^d$, $\bm{g}_{\Phi, jjk}$ is the $k$-th column of the channel matrix $\bm{G}_{\Phi, jj}$,  and $\bm{n}_{jk}\sim\mathcal{CN}(\bm{0},\sigma^2\bm{I}_M)$.
For the full-duplex users, due to the channel reciprocity, we have $\bm{g}_{u,jji}=\bm{g}_{d,jji}$ for $i\in\mathcal{K}_f$. Hence the estimated uplink channels for full-duplex users can also be used for downlink precoding.
\begin{figure}
\centering
 \includegraphics[width=0.4\textwidth,trim = 60mm
      75mm 70mm 76mm, clip]{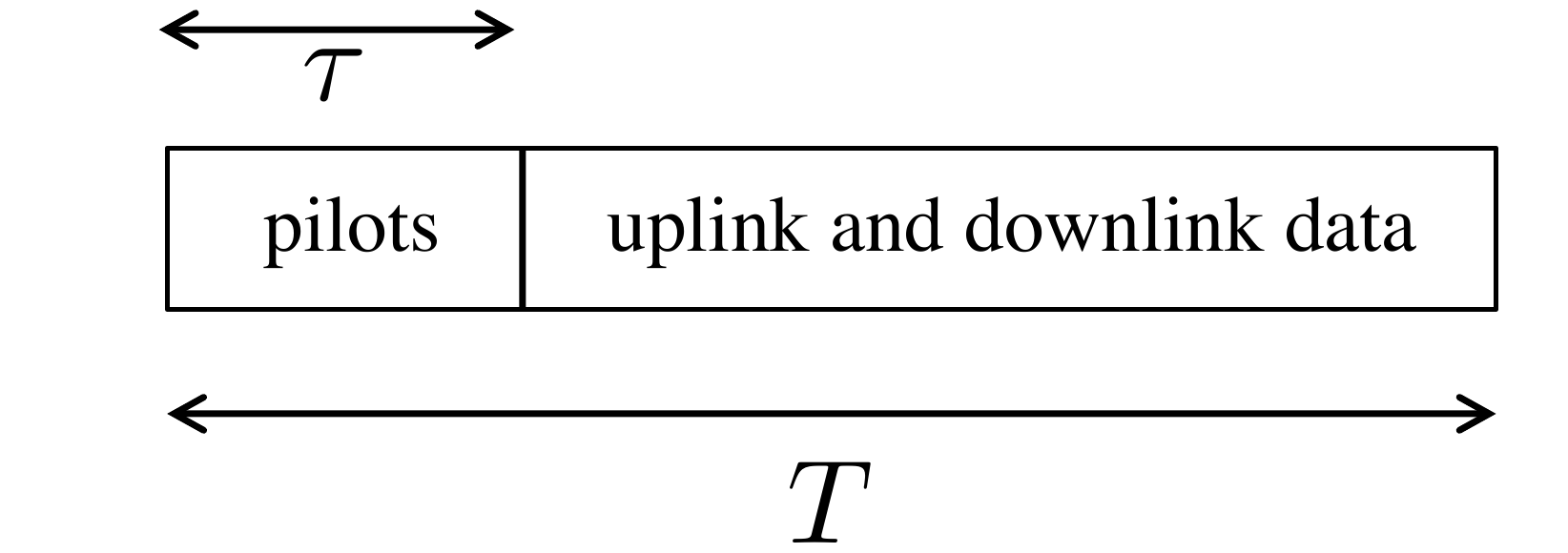}
              \caption{Uplink channel training in full-duplex networks.} \label{training}
    \end{figure}
    
The MMSE channel estimate of the $k$-th user in the $j$-th cell $\bm{g}_{\Phi,jjk}$ can be obtained as~\cite{Hoydis13}
 \begin{gather}
\begin{aligned}
 \hat{\bm{g}}_{\Phi,jjk}=\frac{P_{tr}\beta_{\Phi,jjk}}{\lambda_{\Phi,jk}}\left(\sum_{l=1}^L \bm{g}_{\Phi,jlk}+\frac{\bm{n}_{jk}}{\sqrt{P_{tr}}}\right),\label{eq:mmse}
\end{aligned}
 \end{gather}
where $\lambda_{\Phi,jk}=\sigma^2+P_{tr}\sum_{l=1}^L\beta_{\Phi,jlk}$, and $\hat{\bm{g}}_{\Phi,jjk}\sim\mathcal{CN}\left(\bm{0},\frac{P_{tr}\beta^2_{\Phi,jjk}}{\lambda_{\Phi,jk}}\bm{I}_M\right)$

Due to the orthogonality principle of the MMSE estimator, the true channel can be decomposed as the estimated channel and channel estimation error. Hence we have
 \begin{gather}
\begin{aligned}
\bm{g}_{\Phi,jjk}=\hat{\bm{g}}_{\Phi,jjk}+\bm{\epsilon}_{\Phi,jjk},
\end{aligned}
 \end{gather}
where $\bm{\epsilon}_{\Phi,jjk}\sim\mathcal{CN}\left(\bm{0},\frac{\beta_{\Phi,jjk}\left(\sigma^2+P_{tr}\sum_{l\neq j}\beta_{\Phi,jlk}\right)}{\lambda_{\Phi,jk}}\bm{I}_M\right)$, and the error $\bm{\epsilon}_{\Phi,jjk}$ is independent of the estimate $\hat{\bm{g}}_{\Phi,jjk}$.

\subsubsection{Uplink and downlink data transmission}
In the uplink, the $j$-th BS will apply maximum ratio combining detector by multiplying the conjugate-transpose of channel estimate $\hat{\bm{G}}_{u,jj}^H$ with the received signal 
\begin{eqnarray}
  \bm{r}_{u,j}&=&\hat{\bm{G}}_{u,jj}^H\bm{y}_{u,j}\\
 &=&\hat{\bm{G}}_{u,jj}^H\bm{G}_{u,jj}\bm{x}_{u,j}+\sum_{\mathclap{l\neq j}}\hat{\bm{G}}_{u,jj}^H\bm{G}_{u,jl}\bm{x}_{u,l}+\sum_{l\neq j}\hat{\bm{G}}_{u,jj}^H\bm{V}_{j,l}\bm{x}_{d,l}+\hat{\bm{G}}_{u,jj}^H\bm{z}_{u,j}\\
  &=&\hat{\bm{G}}_{u,jj}^H\hat{\bm{G}}_{u,jj}\bm{x}_{u,j}+\hat{\bm{G}}_{u,jj}^H\bm{E}_{u,jj}\bm{x}_{u,j}+\sum_{\mathclap{l\neq j}}\hat{\bm{G}}_{u,jj}^H\bm{G}_{u,jl}\bm{x}_{u,l}+\sum_{l\neq j}\hat{\bm{G}}_{u,jj}^H\bm{V}_{j,l}\bm{x}_{d,l} \label{eq:ul22}\\
  &&+~\hat{\bm{G}}_{u,jj}^H\bm{z}_{u,j}\nonumber
\end{eqnarray}
where the $k$-th column of $\bm{E}_{u,jj}$ is $\bm{\epsilon}_{u,jjk}$, and (\ref{eq:ul22}) follows from $\bm{G}_{u,jj}=\hat{\bm{G}}_{u,jj}+\bm{E}_{u,jj}$. 

In the downlink, the $l$-th BS will employ conjugate beamforming to precode the downlink messages using the channel estimate $\hat{\bm{G}}_{d,ll}$, and transmit  
an $M\times 1$ signal vector $\bm{x}_{d,l}$, 
\begin{gather}
\begin{aligned}
\bm{x}_{d,l}=\frac{\hat{\bm{G}}^*_{d,ll}}{\sqrt{\tilde{\gamma}_l}}\bm{s}_{d,l}, \label{precode:ip}
\end{aligned}
\end{gather}
where $\tilde{\gamma}_l=\frac{\mathbb{E}(\bm{d}_l^H\hat{\bm{G}}^T_{d,ll}\hat{\bm{G}}^*_{d,ll}\bm{d}_l)}{K_d}=\frac{M}{K_d}\sum_{i=1}^{K_d}\frac{P_{tr}\beta_{d,lli}^2}{\sigma^2+P_{tr}\sum_{j=1}^L\beta_{d,lji}}$.

We assume that the downlink users do not have the channel estimate for reception~(otherwise, additional pilot overhead needs to be considered), but are aware of the channel statistics. Each downlink user can perfectly track the average effective channel gain.
Thus the received downlink signal can be decomposed as an average  effective channel gain times the desired signal symbol, plus a composite term to denote effective noise as in~\cite{Jose11}.

Substituting (\ref{precode:ip}) into (\ref{eq:dlfd}), the received downlink signal at the $k$-th full-duplex user in the $l$-th cell where $k\in\mathcal{K}_f$ can be written as
\begin{gather}
\begin{aligned}
{r}_{d,lk}&=\sqrt{\frac{P_d}{K_d\tilde{\gamma}_l}}\mathbb{E}[\bm{g}^T_{d,llk}\bm{v}_{lk}]d_{l,k}+\sqrt{\frac{P_d}{K_d\tilde{\gamma}_l}}\left(\bm{g}^T_{d,llk}\bm{v}_{lk}-\mathbb{E}[\bm{g}^T_{d,llk}\bm{v}_{lk}]\right)d_{l,k}\\
&+\sum_{(j,i)\neq (l,k)}\sqrt{\frac{P_d}{K_d\tilde{\gamma}_j}}\bm{g}^T_{d,jlk}{\bm{v}}_{ji}d_{j,i}+\sum_{j=1}^L\sum_{n=1}^{K_u}\sqrt{P_u}f_{lkjn}{u}_{j,n}-\sqrt{P_u}f_{lklk}{u}_{l,k}+z_{d,lk}.
\end{aligned}
\end{gather}
 where $\bm{v}_{lk}=\hat{\bm{g}}^*_{d,llk}$ for conjugate beamforming precoder.

Similarly, we substitute (\ref{precode:ip}) into (\ref{eq:dlhd}) to obtain the downlink signal at the $k^\prime$-th half-duplex user in the $l$-th cell where $k^\prime\in\mathcal{K}_h^d$ 
\begin{gather}
\begin{aligned}
{r}_{d,lk^\prime}&=\sqrt{\frac{P_d}{K_d\tilde{\gamma}_l}}\mathbb{E}[\bm{g}^T_{d,llk^\prime}\bm{v}_{lk^\prime}]d_{l,k^\prime}+\sqrt{\frac{P_d}{K_d\tilde{\gamma}_l}}\left(\bm{g}^T_{d,llk^\prime}\bm{v}_{lk^\prime}-\mathbb{E}[\bm{g}^T_{d,llk^\prime}\bm{v}_{lk^\prime}]\right)d_{l,k^\prime}\\
&+\sum_{(j,i)\neq (l,k^\prime)}\sqrt{\frac{P_d}{K_d\tilde{\gamma}_j}}\bm{g}^T_{d,jlk^\prime}{\bm{v}}_{ji}d_{j,i}+\sum_{\mathclap{\substack{j=1}}}^L\sum_{\mathclap{n=1}}^{K_u} \sqrt{P_u}f_{lk^\prime jn}u_{j,n}+{n}_{d,lk^\prime}.
\end{aligned}
\end{gather}

\subsubsection{Ergodic achievable rates}
For the uplink, the channel estimate will be treated as the true channel.
Hence we can obtain the ergodic achievable rate for the $n$-th uplink user in the $j$-th cell as follows:
\begin{gather}
\begin{aligned}
\tilde{R}^{fd,ip}_{u,jn}&= \mathbb{E}\left\{\mathrm{log}_2\left(1+\frac{P_u\norm{\hat{\bm{g}}_{u,jjn}}^4}{P_u|\hat{\bm{g}}_{u,jjn}^H\bm{\epsilon}_{u,jjn}|^2+P_u\sum_{(l,m)\neq (j,n)}|\hat{\bm{g}}^H_{u,jjn}\bm{g}_{u,jlm}|^2+I^{ip}_{bs-bs}+N}\right)\right\}, \label{ulim:lb1}
\end{aligned}
\end{gather}
where $n\in\mathcal{K}_u$, $I^{ip}_{bs-bs}=\sum_{l\neq j}\sum_{k\in\mathcal{K}_d}\frac{P_d}{K_d \tilde{\gamma}_l}|\hat{\bm{g}}^H_{u,jjn}\bm{V}_{jl}\hat{\bm{g}}^*_{d,llk}|^2,~N=\norm{\hat{\bm{g}}_{u,jjn}}^2(\sigma^2+\kappa P_d\beta_{b,jj})$. 

For the downlink, the effective noise is uncorrelated with the signal, using worst-case independent Gaussian noise results in~\cite{Jose11}, the downlink ergodic achievable rates of the $k$-th full-duplex user and the $k^\prime$-th half-duplex user in the $l$-th cell are respectively given as
\begin{gather}
\begin{aligned}
{R}_{d,lk}^{fd,ip}&= \mathrm{log}_2\left(1+\frac{\frac{P_d}{K_d\tilde{\gamma}_l}\left|\mathbb{E}\left[\bm{g}^T_{d,llk}\hat{\bm{g}}^*_{d,llk}\right]\right|^2}{\frac{P_d}{K_d\tilde{\gamma}_l}\mathsf{var}\left[\bm{g}^T_{d,llk}\hat{\bm{g}}^*_{d,llk}\right]+I_{other}(k)- P_u\mathbb{E}\left[\abs{f_{lklk}}^2\right]+\sigma^2+\kappa P_u\beta_{I,lklk}}\right),\\
{R}_{d,lk^\prime}^{fd,ip}&= \mathrm{log}_2\left(1+\frac{\frac{P_d}{K_d\tilde{\gamma}_l}\left|\mathbb{E}\left[\bm{g}^T_{d,llk^\prime}\hat{\bm{g}}^*_{d,llk^\prime}\right]\right|^2}{\frac{P_d}{K_d\tilde{\gamma}_l}\mathsf{var}\left[\bm{g}^T_{d,llk^\prime}\hat{\bm{g}}^*_{d,llk^\prime}\right]+I_{other}(k^\prime)+\sigma^2}\right),\label{dlim:lb1}
\end{aligned}
\end{gather}
where $k\in\mathcal{K}_f,k^\prime\in\mathcal{K}_h^d$. $I_{other}(k)=\sum_{(j,i)\neq (l,k)}\frac{P_d}{K_d\tilde{\gamma}_j}\mathbb{E}\left[\abs{\bm{g}^T_{d,jlk}\hat{\bm{g}}^*_{d,jji}}^2\right]+\sum_{j=1}^L\sum_{n\in\mathcal{K}_u}P_u\mathbb{E}\left[\abs{f_{lkjn}}^2\right]$. The notation $\mathsf{var}[x]\triangleq\mathbb{E}[(x -\mu)(x -\mu)^H],~\mathbb{E}[x]=\mu$.

\begin{Proposition} \label{prop2}
For imperfect CSI with MMSE estimation, the following sets of rates are achievable in multi-cell MU-MIMO full-duplex networks when $M\geq 3$ 
\begin{gather}
\begin{aligned}
\text{Uplink user:}~R^{fd,ip}_{u,jn}&= \mathrm{log}_2\bigg(1+\frac{P_{tr}P_u(M-1)\beta^2_{u,jjn}}{P_u\beta_{u,jjn}(\sigma^2+P_{tr}\sum_{l\neq j}\beta_{u,jln})+\tilde{I}_{up}+\tilde{N}}\bigg),\\
\text{Downlink, FD user:}~R^{fd,ip}_{d,lk}&= \mathrm{log}_2\Bigg(1+\frac{\tilde{\eta}_lP_{tr} P_dM\beta^4_{d,llk}}{\lambda^2_{d,lk}(\tilde{I}_{down}(k)-P_u\beta_{I,lklk}+\sigma^2+\kappa P_u\beta_{I,lklk})}\Bigg),\\
\text{Downlink, HD user:}~R^{fd,ip}_{d,lk^\prime}&= \mathrm{log}_2\Bigg(1+\frac{\tilde{\eta}_lP_{tr} P_dM\beta^4_{d,llk^\prime}}{\lambda^2_{d,lk^\prime}(\tilde{I}_{down}(k^\prime)+\sigma^2)}\Bigg),\label{eq:prop2}
\end{aligned}
\end{gather}
where  $\tilde{I}_{up}=P_{tr}P_u\sum_{l\neq j}\left[(M+1)\beta^2_{u,jln}+\sum_{l_1\neq l}\beta_{u,jl_1 n}\beta_{u,jln}+\frac{\beta_{u,jln}\sigma^2}{P_{tr}}\right]$, $\tilde{N}=\lambda_{u,jn}(P_d\sum_{l\neq j} \beta_{b,jl}+P_u\sum_{l=1}^L\sum_{m\in\mathcal{K}_u,m\neq n}\beta_{u,jlm}+\sigma^2+\kappa P_d\beta_{b,jj})$, $\tilde{I}_{down}(k) = \frac{\tilde{\eta}_l P_d\beta^3_{d,llk}}{\lambda_{d,lk}}+\sum_{j\neq l}\frac{\tilde{\eta}_jP_{tr}P_d(M+1)\beta^2_{d,jlk}\beta^2_{d,jjk}}{\lambda^2_{d,jk}}+\sum_{j\neq l}\frac{\tilde{\eta}_j P_{d}\beta_{d,jjk}^2(\sigma^2+P_{tr}\sum_{l_1\neq l}\beta_{d,jl_1k})\beta_{d,jlk}}{\lambda^2_{d,jk}}+\sum_{j=1}^L\sum_{i\neq k,i\in\mathcal{K}_d}\frac{\tilde{\eta}_jP_d\beta^2_{d,jji}\beta_{d,jlk}}{\lambda_{d,ji}}+\sum_{j=1}^L\sum_{n\in\mathcal{K}_u}P_u\beta_{I,lk jn}$. $\tilde{\eta}_l=\left(\sum_{i\in\mathcal{K}_d}\frac{\beta^2_{d,lli}}{\lambda_{d,li}}\right)^{-1}$, $n\in\mathcal{K}_u,~k\in\mathcal{K}_f,~k^\prime\in\mathcal{K}_h^d.$ 
\end{Proposition}
\begin{proof}
See Appendix~\ref{proof2}.
\end{proof}

 \subsection{TDD baseline system }
We use TDD system as a baseline half-duplex system for comparison since both full-duplex and TDD systems use uplink training for channel estimation.\footnote{FDD system requires downlink training with channel feedback which incurs a much higher pilot overhead as the overhead is not only proportional to the number of users but also the number of BS antennas.} Compared with a full-duplex system, in a TDD system, the uplink does not have BS-BS interference and transmit noise which accounts for imperfect FD radios, and the downlink does not have UE-UE interference and transmit noise. We assume the corresponding TDD system has an uplink and downlink user sets of $\mathcal{K}_u$ and $\mathcal{K}_d$, respectively, with the same total number of uplink and downlink users as in the full-duplex system, i.e., $|\mathcal{K}_u|=K_u$ and $|\mathcal{K}_d|=K_d$. For the TDD baseline system, the up- and downlink transmissions are in two different time slots, and we assume an equal time sharing between up- and downlink transmission.

By treating interference as noise, the up- and downlink ergodic achievable rates in a TDD system under perfect CSI assumption are given as 
\begin{gather}
\begin{aligned}
\text{Uplink user:}~{R}^{tdd,p}_{u,jn}&=\frac{1}{2}\mathbb{E}\left\{\mathrm{log}_2\left(1+\frac{P_u\norm{\bm{g}_{u,jjn}}^4}{P_u\sum_{(l,m)\neq (j,n)}|\bm{g}^H_{u,jjn}\bm{g}_{u,jlm}|^2+\norm{\bm{g}_{u,jjn}}^2\sigma^2}\right)\right\},\\
\text{Downlink user:}{R}^{tdd,p}_{d,lk}&= \frac{1}{2} \mathbb{E}\left\{\mathrm{log}_2\left(1+\frac{\frac{P_d}{K_d\gamma_l}\norm{\bm{g}_{d,llk}}^4}{\sum_{(j,i)\neq (l,k)}\frac{P_d}{K_d\gamma_j}|\bm{g}^T_{d,jlk}\bm{g}^*_{d,jji}|^2+\sigma^2}\right)\right\},
\label{ARHDP}
\end{aligned}
\end{gather}
where $\gamma_l=\frac{{M\sum_{k=1}^{K_d} \beta_{d,llk}}}{K_d}$, $n\in\mathcal{K}_u,~k\in\mathcal{K}_d$. 

When training sequences are used for channel estimation, within a coherence interval $T$, $K_u$ mutually orthogonal pilot sequences of length $\tau_u$ ($\tau_u\geq K_u$) are used for the uplink users and $K_d$ mutually orthogonal pilot sequences of length $\tau_d$ ($\tau_d\geq K_d$) are used for the downlink users. The same sets of training sequences are also reused by all cells. The corresponding up- and downlink ergodic achievable rates with MMSE estimation in TDD system are given below,
\begin{gather}
\begin{aligned}
\text{Uplink user:}~ {R}^{tdd,ip}_{u,jn}&=\frac{1}{2} \mathbb{E}\left\{\mathrm{log}_2\left(1+\frac{P_u\norm{\hat{\bm{g}}_{u,jjn}}^4}{P_u|\hat{\bm{g}}_{u,jjn}^H\bm{\epsilon}_{u,jjn}|^2+P_u\sum_{(l,m)\neq (j,n)}|\hat{\bm{g}}^H_{u,jjn}\bm{g}_{u,jlm}|^2+N^\prime}\right)\right\},\\
\text{Downlink user:}~{R}^{tdd,ip}_{d,lk}&= \frac{1}{2}  \mathbb{E}\left\{\mathrm{log}_2\left(1+\frac{\frac{P_d}{K_d\tilde{\gamma}_l}\left|\mathbb{E}\left[\bm{g}^T_{d,llk}\hat{\bm{g}}^*_{d,llk}\right]\right|^2}{Z^\prime+\sigma^2}\right)\right\},\label{ARHDIP}
\end{aligned}
\end{gather}
where $N^\prime=\norm{\hat{\bm{g}}_{u,jjn}}^2\sigma^2$, $n\in\mathcal{K}_u,~k\in\mathcal{K}_d$, and $Z^\prime=\frac{P_d}{K_d\tilde{\gamma}_l}|\bm{g}^T_{d,llk}\hat{\bm{g}}^*_{d,llk}-\mathbb{E}[\bm{g}^T_{d,llk}\hat{\bm{g}}^*_{d,llk}]|^2+\sum_{(j,i)\neq (l,k)}\frac{P_d}{K_d\tilde{\gamma}_j}|\bm{g}^T_{d,jlk}{\bm{v}}_{ji}d_{j,i}|^2$.

Note that the uplink and downlink ergodic achievable rates in the baseline TDD system follow~\cite{Ngo13} and \cite{Jose11} but without lower bounding the achievable rates, as we will use them to compute the full-duplex versus half-duplex rate ratios in the next section.

\section{Large-Scale full-duplex System Performance} \label{asym}
While the general ergodic achievable rates in the full-duplex system are given in the previous section, it is of interest to study the impact of large antenna arrays at BSs as the next generation of wireless systems will employ significantly more  antennas at the infrastructure nodes~\cite{3gppMassiveMIMO}. In this section,
we will investigate large-scale system performance as the number of full-duplex BS antennas, $M$, becomes arbitrarily large. 
\subsection{Leveraging large antenna arrays for multi-cell interference mitigation}
In this section, we will show that using large BS antenna arrays can mitigate multi-cell interference in the full-duplex system.
With increasing number of BS antennas, the signal strength will become stronger due to beamforming, and the transmit power can be scaled down proportionally to maintain the same quality-of-service. In what follows, we will present two theorems which characterize the asymptotic full-duplex spectral efficiency gain over TDD system under both perfect and imperfect CSI assumptions. 
\begin{theorem}[Asymptotic FD spectral efficiency gain with perfect CSI] \label{prop3}
For perfect CSI, we scale the transmit power of each node proportional to $1/M$ as $P_u=E_u/M$ and $P_d=E_d/M$, where $E_u$ and $E_d$ are fixed. As $M\rightarrow\infty$, the full-duplex spectral efficiency gains over the TDD system in the uplink and downlink, denoted by $\mathsf{Gain}^p_{u}$ and $\mathsf{Gain}^p_{d}$, respectively,  are given below, where fixed asymptotic up- and downlink rates are maintained. 
\begin{gather}
\begin{aligned}
\mathsf{Gain}^{p}_{u}&\triangleq\lim_{M\rightarrow\infty}\frac{\sum_{j=1}^L\sum_{n\in\mathcal{K}_u}R^{fd,p}_{u,jn}}{\sum_{j=1}^L\sum_{n\in\mathcal{K}_u}R_{u,jn}^{tdd,p}}=2,\\
\mathsf{Gain}^p_{d}&\triangleq\lim_{M\rightarrow\infty}\frac{\sum_{l=1}^L\sum_{k\in\mathcal{K}_d}R^{fd,p}_{d,lk}}{\sum_{l=1}^L\sum_{k\in\mathcal{K}_d}R_{d,lk}^{tdd,p}}=2.
\end{aligned}
\end{gather}
The asymptotic ergodic achievable rate of the $n$-th uplink user in the $j$-th cell and the achievable rate of the $k$-th downlink user in the $l$-th cell are given below,
\begin{gather}
\begin{aligned}
R^{fd,p}_{u,jn}&\rightarrow \mathrm{log}_2\left(1+\frac{\beta_{u,jjn}E_u}{\sigma^2}\right), ~n\in\mathcal{K}_u\\
R^{fd,p}_{d,lk}&\rightarrow \mathrm{log}_2\left(1+\frac{\beta^2_{d,llk}E_d}{\sum_{i\in\mathcal{K}_d}\beta_{d,lli}\sigma^2}\right), ~k\in\mathcal{K}_d.
\end{aligned}
\end{gather}
\end{theorem}

\begin{proof}
For two mutually independent $M\times 1$ vectors $\boldsymbol{a}\triangleq [a_1,\cdots,a_M]^T$ and $\boldsymbol{b}\triangleq [b_1,\cdots,b_M]^T$ whose entries are i.i.d. zero-mean random variables with $\mathbb{E}(|a_i|^2=\sigma^2_a)$ and $\mathbb{E}(|b_i|^2=\sigma^2_b),~\forall i\in\{1,\cdots,M\}$, by law of large numbers, we have the following almost sure convergence according to~\cite{Cramer1970},
\begin{gather}
\begin{aligned}
\frac{1}{M}\bm{a}^H\bm{a}\xrightarrow[]{a.s.} \sigma^2_a,~\frac{1}{M}\bm{a}^H\bm{b}\xrightarrow[]{a.s.} 0,~\text{as}~M\rightarrow\infty.\label{lrn1}
\end{aligned}
\end{gather}

Under the favorable propagation condition in \cite{marzettaNon10} where the fast fading channels are i.i.d. with zero mean and unit variance, from (\ref{lrn1}), in the limit of $M$,  we have $\lim_{M\rightarrow\infty}\frac{\bm{G}_{u,jl}^H\bm{G}_{u,jl}}{M}=\bm{D}_{u,jl}\delta_{jl}$, $\lim_{M\rightarrow\infty}\frac{\bm{G}_{d,jl}^H\bm{G}_{d,jl}}{M}=\bm{D}_{d,jl}\delta_{jl}$, where $\bm{D}_{u,jl}$ is a $K_u\times K_u$ diagonal matrix, and each diagonal element is $(\bm{D}_{u,jl})_{n}=\beta_{u,jln}$; $\bm{D}_{d,jl}$ is a $K_d\times K_d$ diagonal matrix, and each diagonal element is $(\bm{D}_{d,jl})_{k}=\beta_{d,jlk}$; $\delta_{jl}=1$~for $j=l$, $\delta_{jl}=0$~for $j\neq l$. 
By substituting $P_u=E_u/M$ and $P_d=E_d/M$ into (\ref{ul:lb1}), (\ref{ARFDP}) and (\ref{ARHDP}), we can obtain the desired results as $M\rightarrow\infty$. 
\end{proof}
\begin{remark}
When a full-duplex BS employs a large antenna array, the transmit power of each node can possibly be scaled down proportionally to $1/M^{C}$ to achieve the same rate. As $M\rightarrow\infty$, $1/M$ is the fastest rate at which we can scale down the transmit power to maintain fixed asymptotic up- and downlink rates. Otherwise,  if $C>1$, the rates will go to zero and if $C<1$, the rates will go to infinity. The full-duplex system preserves the same power scaling law as in the half-duplex system~\cite{Ngo13} despite increased interference in the multi-cell full-duplex networks. Full-duplex system asymptotically doubles the spectral efficiency over TDD system since all $K_u$ uplink streams and $K_d$ downlink streams can be supported in the same time-frequency slot. 
\end{remark}

\begin{theorem}[Asymptotic FD spectral efficiency gain with imperfect CSI]\label{prop4}
For imperfect CSI with MMSE estimation, we scale the power of each node for channel training and data  transmission proportional to $1/\sqrt{M}$ as $P_{tr}=E_{tr}/{\sqrt{M}}$, $P_u=E_u/{\sqrt{M}}$ and $P_d={E_d}/{\sqrt{M}}$, where $E_{tr}$, $E_u$ and $E_d$ are fixed. Within a coherence interval $T$, let $\tau_u=K_u$, $\tau_d=K_d$ and $\tau=K_u+K_d-K_f$.
As $M\rightarrow\infty$, the full-duplex spectral efficiency gains over the TDD system in the uplink and downlink, denoted by $\mathsf{Gain}^{ip}_{u}$ and $\mathsf{Gain}^{ip}_{d}$, respectively,  are given below, where fixed asymptotic up- and downlink rates are maintained. 
\begin{gather}
\begin{aligned}
\mathsf{Gain}^{ip}_{u}&\triangleq\lim_{M\rightarrow\infty}\frac{\frac{T-\tau}{T}\sum_{j=1}^L\sum_{n\in\mathcal{K}_u}R^{fd,ip}_{u,jn}}{\frac{T-\tau_u}{T}\sum_{j=1}^L\sum_{n\in\mathcal{K}_u}R_{u,jn}^{tdd,ip}}=2\Big(1-\frac{K_h^d}{T-K_u}\Big),\\
\mathsf{Gain}^{ip}_{d}&\triangleq\lim_{M\rightarrow\infty}\frac{\frac{T-\tau}{T}\sum_{l=1}^L\sum_{k\in\mathcal{K}_d}R^{fd,ip}_{d,lk}}{\frac{T-\tau_d}{T}\sum_{l=1}^L\sum_{k\in\mathcal{K}_d}R_{d,lk}^{tdd,ip}}=2\Big(1-\frac{K_h^u}{T-K_d}\Big).
\end{aligned}
\end{gather}
The asymptotic ergodic achievable rate of the $n$-th uplink user in the $j$-th cell and the achievable rate of the $k$-th downlink user in the $l$-th cell are given below,
\begin{gather}
\begin{aligned}
R_{u,jn}^{fd,ip}&\rightarrow \mathrm{log}_2\left(1+\frac{E_{tr}E_u\beta^2_{u,jjn}}{E_{tr} E_u\sum_{l\neq j} \beta^2_{u,jln}+\sigma^4}\right),\\
R_{d,lk}^{fd,ip}&\rightarrow \mathrm{log}_2\left(1+\frac{E_{tr}E_d \beta^4_{d,llk}}{Z_l\left(\sum_{j\neq l} \frac{E_{tr}E_d\beta^2_{d,jjk}\beta^2_{d,jlk}}{Z_j}+\sigma^6\right)}\right), 
\end{aligned}
\end{gather}
where $Z_l=\sum_{i\in\mathcal{K}_d}\frac{\beta^2_{d,lli}}{\sigma^2},~n\in\mathcal{K}_u,~k\in\mathcal{K}_d$.
\end{theorem}
\begin{proof}
When channel estimation overhead is taken into account, the up- and downlink spectral efficiency in the full-duplex system are $\frac{T-\tau}{T}\sum_{j=1}^L\sum_{n\in\mathcal{K}_u}R^{ip}_{u,jn}$ and $\frac{T-\tau}{T}\sum_{l=1}^L\sum_{k\in\mathcal{K}_d}R^{ip}_{d,lk}$, respectively. While for the TDD system, since uplink and downlink transmission are in two different time slots, the corresponding up- and downlink spectral efficiency are $\frac{T-\tau_u}{T}\sum_{j=1}^L\sum_{n\in\mathcal{K}_u}R_{u,jn}^{tdd,ip}$ and $\frac{T-\tau_d}{T}\sum_{l=1}^L\sum_{k\in\mathcal{K}_d}R_{d,lk}^{tdd,ip}$, respectively. 
Following similar steps used in the proof of Theorem~\ref{prop3}, substituting $P_{tr}=E_{tr}/{\sqrt{M}}$, $P_u=E_u/{\sqrt{M}}$, $P_d={E_d}/{\sqrt{M}}$, $\tau_u=K_u$, $\tau_d=K_d$ and $\tau=K_u+K_d-K_f$ into (\ref{ulim:lb1}), (\ref{dlim:lb1}) and (\ref{ARHDIP}), and using the fact that $K_u=K_h^u+K_f$ and $K_d=K_h^d+K_f$, the desired results can be obtained as $M\rightarrow\infty$.
\end{proof}
\begin{corollary}
When $K_u=K_d=K_f$, i.e., when only full-duplex users are served, full-duplex system asymptotically doubles the spectral efficiency over the TDD system under the imperfect CSI assumption, where $\mathsf{Gain}^{ip}_{u}=\mathsf{Gain}^{ip}_{d}=2$. However, if there exists half-duplex users~(i.e., $K_h^d>0$ or $K_h^u>0$), then $\mathsf{Gain}^{ip}_{u/d}<2$, and the spectral efficiency gains decrease as the number of half-duplex users increases due to channel training overhead.
\end{corollary}
\begin{remark}
In case of imperfect CSI, the fastest rate at which we can cut down the transmit power to maintain fixed asymptotic up- and downlink rates is $1/\sqrt{M}$.
As the transmit power scale down with increasing $M$, the impact of imperfect SI cancellation, intra-cell and inter-cell interference in the multi-cell MU-MIMO full-duplex networks will vanish as $M\rightarrow\infty$.
\end{remark}

\subsection{System with only full-duplex users}
For tractability, we consider a homogeneous network by assuming all links in the same cell have the same channel statistic and all the cross-cell links have the same channel statistics, i.e., let $\beta_{u,lln}=\beta_{d,llk}=\beta_{I,lkln}=\beta_{b,ll}=1$, $\beta_{u,jln}=\beta_{d,jlk}=\beta_{I,lkjn}=\beta_{b,jl}=\beta$, where $\beta\in[0,1]$ for $j\neq l\in[1,\cdots,L],~n\in\mathcal{K}_u$, $k\in\mathcal{K}_d$, and the noise variance is set as $\sigma^2=1$. We consider a simple scenario where all users are full-duplex and the number of full-duplex users in each cell is $K_f\triangleq K$. Similar results can be derived in case of mixed half-duplex UEs or only half-duplex UEs. With such simplification, we can compute the up- and downlink spectral efficiency per cell (bits/s/Hz/cell) when CSI is perfect as
\begin{gather}
\begin{aligned}
R^{fd,p}_{u}&=K\mathrm{log}_2\left(1+\frac{P_u(M-1)}{P_u(K-1)+(L-1)\beta(P_uK+P_d)+\kappa P_d+1}\right),\\
R^{fd,p}_{d}&=K\mathrm{log}_2\left(1+\frac{P_d(M-1)(M-2)}{P_d(K-1)(M-2)+MK(L-1)\beta P_d+V+MK(\kappa P_u+1)}\right),
\label{homocsi}
\end{aligned}
\end{gather}
where $V=MK\left(K-1+(L-1)K\beta\right)P_u$.

When CSI is imperfect, the up- and downlink spectral efficiency per cell (bits/s/Hz/cell) are 
\begin{gather}
\begin{aligned}
R^{fd,ip}_{u}&\!=\frac{K(T-\tau)}{T}\mathrm{log}_2\left(\!1+\frac{P_{tr} P_u(M-1)}{P_{tr} P_u\left(K\bar{L}^2-1+\beta(\bar{L}-1)M\right)+J}\right),\\
R^{fd,ip}_{d}&\!=\frac{K(T-\tau)}{T}\mathrm{log}_2\left(\!1+\frac{P_{tr}P_dM}{\left(1+P_{tr}\bar{L}\right)U_1+\left(\bar{L}-1\right)U_2}\right),\label{homoerror}
\end{aligned}
\end{gather}
where $J=P_d(1+P_{tr}\bar{L})(\bar{L}-1)+P_uK\bar{L}+P_{tr}(1+ \kappa P_d)\bar{L}+\kappa P_d+1$,\\
$U_1=P_d\left(1+(K-1)\bar{L}\right)+K\left(P_u(K-1)+P_u(\bar{L}-1)K+1+\kappa P_u\right),
$\\
$U_2=P_{tr}P_d\left(M\beta+\bar{L}\right)+P_d,~\bar{L}=1+(L-1)\beta$.

We first evaluate the tightness of our derived bounds on achievable rates in the homogenous networks.
We consider a scenario with $L=7$ cells and the inter-cell interference level $\beta=0.3$, each cell has $K=5$ full-duplex UEs.
The up- and downlink transmit power are assumed as $P_{tr}=P_u=10$~dB and $P_d=20$~dB, respectively. The dynamic range parameter is $\kappa=-50$~dB. In case of imperfect CSI, considering an OFDM system, we assume the coherence time is 1~ms (one subframe in LTE standard where there are 14 OFDM symbols in each subframe), and the ``frequency smoothness interval" is 14 as given in~\cite{marzettaNon10}. Hence the coherence interval which is a time-frequency product is equal to $T=196$. The pilot length is assumed to be the same as the number of users, i.e., $\tau=K$. From Figure~\ref{tight}, we can see that all bounds are very tight in both perfect and imperfect CSI cases, particularly with increasing $M$. Next, we use these bounds for the full-duplex system to compute the uplink and downlink spectral efficiency ratios between full-duplex system and half-duplex system. Note that for the half-duplex system, we still numerically evaluate the ergodic achievable rates with no simplification.
\begin{figure}
    \centering
    \begin{subfigure}[b]{0.65\textwidth}
        \centering
        \includegraphics[width=0.72\textwidth,trim = 10mm
      60mm 15mm 70mm, clip]{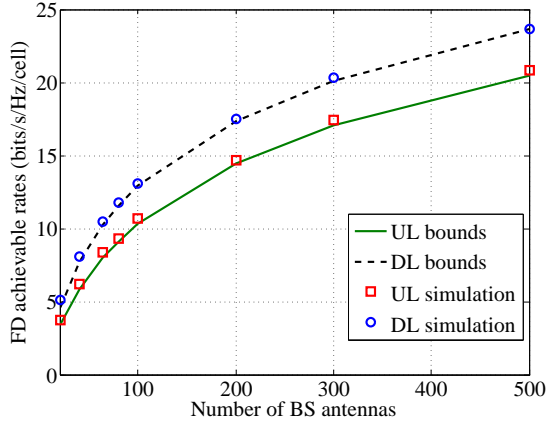}
              \caption{With perfect CSI} 
              \label{tight1}
    \end{subfigure}%
    ~ \hspace{-30mm}
    \begin{subfigure}[b]{0.65\textwidth}
        \centering
        \includegraphics[width=0.68\textwidth,trim = 20mm
      60mm 15mm 70mm, clip]{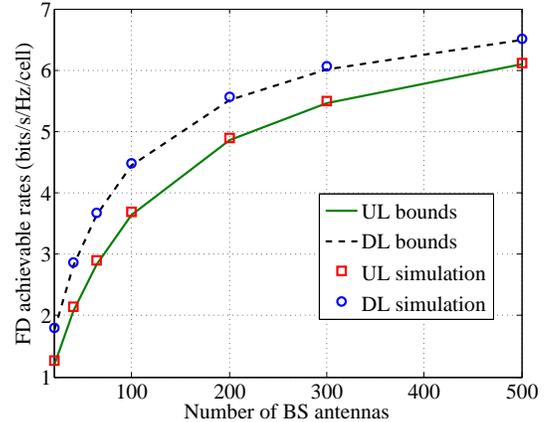}
              \caption{With imperfect CSI} \label{tight2}
    \end{subfigure}
    \caption{Comparisons between lower bounds in Proposition~\ref{prop1} and \ref{prop2} and numerically evaluated values of the ergodic achievable rates~(bits/s/Hz/cell) under both perfect and imperfect CSI assumptions, where $P_{tr}=P_u=10$~dB and $P_d=20$~dB.}\label{tight}
\end{figure}

We compute the spectral efficiency ratios between full-duplex and half-duplex by comparing the full-duplex rates given in (\ref{homocsi}) and (\ref{homoerror}) with the half-duplex achievable rates in (\ref{ARHDP}) and (\ref{ARHDIP}) which can be evaluated numerically. We consider the same setting as in Fig.~\ref{tight} under both perfect and imperfect CSI assumptions.
In Figure~\ref{power1}, in the case of perfect CSI, we can see that as $M$ increases,  the up- and downlink spectral efficiency gain between full-duplex system versus TDD system will converge to 2 as we scale down the transmit power according to Theorem~\ref{prop3}. The convergence rate is fast at the beginning, and becomes slower for large $M$. To reach the asymptotic $2\times$ gain, it will require remarkably large number of BS antennas. However, full-duplex spectral efficiency gains in both uplink and downlink are achievable for finite $M$. For example, when $M=64$, full-duplex achieves about 1.7$\times$ downlink gain and 1.3$\times$ uplink gain. We numerically show that even without scaling down the transmit power, similar full-duplex gains are achievable. Figure~\ref{power2} shows the full-duplex over half-duplex spectral efficiency gains in the case of imperfect CSI. We observe that even with channel estimation error and pilot contamination, full-duplex uplink and downlink gains exist for finite $M$ with and without scaling down the transmit power.
\begin{figure}
    \centering
    \begin{subfigure}[b]{0.6\textwidth}
        \centering
        \includegraphics[width=0.72\textwidth,trim = 10mm
      60mm 15mm 65mm, clip]{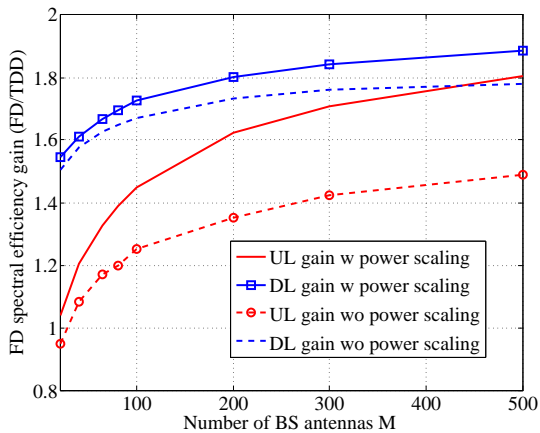}
              \caption{With perfect CSI} \label{power1}
    \end{subfigure}%
    ~ \hspace{-25mm}
    \begin{subfigure}[b]{0.65\textwidth}
        \centering
        \includegraphics[width=0.7\textwidth,trim = 10mm
      63mm 15mm 65mm, clip]{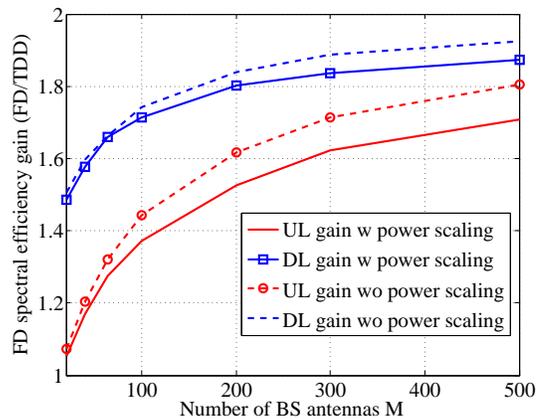}
              \caption{With imperfect CSI} \label{power2}
    \end{subfigure}
    \caption{Spectral efficiency gains of using full-duplex for finite $M$ with and without power scaling.}\label{power scaling}
\end{figure}

We also investigate the spectral efficiency gain and antenna reduction tradeoff between full-duplex system and TDD system. The antenna reduction is the reduction of BS antennas due to full-duplex operation which can be characterized by the ratio between the number of BS antennas in TDD system~($M_{\rm{TDD}}$) and the number of BS antennas in FD system. The spectral efficiency gain and antenna reduction tradeoff is essentially the tradeoff between full-duplex multiplexing gain due to simultaneous transmission and reception and beamforming gain due to large $M$. In Fig.~\ref{ratio}, we consider the same setting as in Fig.~\ref{tight} under imperfect CSI assumption. We illustrate the spectral efficiency gain and antenna reduction tradeoff for both uplink and downlink with different $M_{\rm{TDD}}$ in the TDD system. Larger full-duplex antenna reduction can be achieved at the cost of reducing the spectral efficiency gain. In the regimes above the dashed arrows as shown in Fig.~\ref{ratio}, full-duplex system achieves both spectral efficiency gain and antenna reduction over TDD system. We can see that full-duplex system can require an order of magnitude fewer BS antennas compared with TDD to obtain the same performance in some cases. In addition, as $M_{\rm{TDD}}$ increases, full-duplex system can achieve higher spectral efficiency gain and antenna reduction. 
\begin{figure*}[t!]
        \centering
          \includegraphics[width=0.6\textwidth,trim = 3mm
      60mm 0mm 65mm, clip]{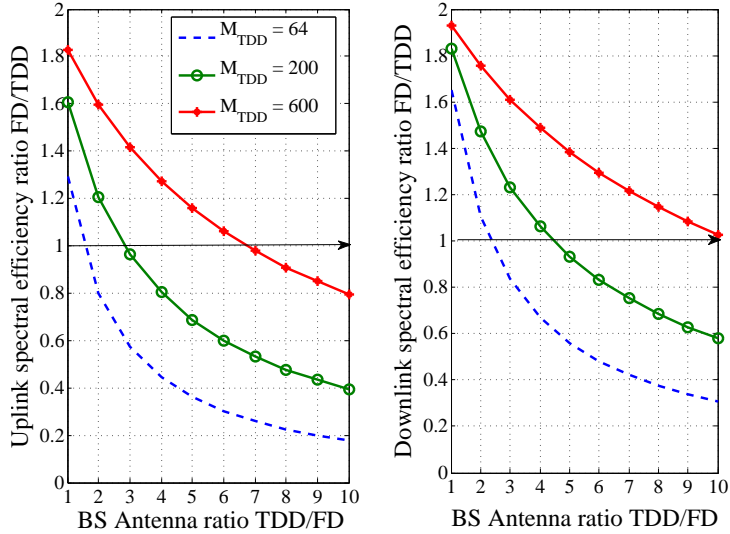}
    \caption{Spectral efficiency gain and antenna reduction tradeoff for various numbers of BS antennas employed in the TDD system under imperfect CSI assumption.}\label{ratio}
\end{figure*}

\section{Numerical Results}\label{numerical}

\subsection{A small cell scenario}
Based on the state-of-art self-interference cancellation capability~\cite{ashu13},  the coverage of a full-duplex system is more likely to be within a small-cell communication ranges. Hence in this section, we present the numerical simulation in realistic small-cell network settings used in 3GPP~\cite{3gppdtdd} to evaluate the system performance.
Twelve small cell BSs are uniformly and randomly distributed within a hexagonal region with a radius of $300$~meters 
All the small-cell BSs have full-duplex capability with multiple antennas. Each small-cell BS is associated with five single-antenna half-duplex uplink UEs and downlink UEs, respectively, which are uniformly and randomly dropped within a radius of 40 meters of the BS. The numerical results are shown for the case of imperfect CSI with channel estimation error. We consider an OFDM system and the coherence interval is $T=196$.
The channel bandwidth is assumed as 20 MHz for both TDD and full-duplex systems. 
The large-scale fading models for BS-UE, BS-BS and UE-UE channels which include path loss and shadowing effect follow 3GPP model in~\cite{3gppdtdd}. The SI channel
model is based on the existing experiment data~\cite{ashu13}, where the propagation loss of SI channel
in a separate-antenna system includes path loss, isolation, cross-polarization and antenna directionality~\cite{evan13}, and in a shared-antenna system includes isolation using a circulator. We assume the SI channel has a propagation loss of $40$~dB. 
We run hundreds of random drops of BSs and UEs in the simulation, and parameter details are given in Tabel~\ref{simulation:table}.
\begin{table}
\centering
\renewcommand{\arraystretch}{1.2}
\scalebox{0.9}{
\begin{tabular}{|l|l|} \hline
\textbf{Parameter}& \textbf{Value}\\ \hline
Full-duplex BS power &$24$~dBm\\ \hline
UE power &$23$~dBm\\ \hline
BS antenna gain &$5$~dBi\\ \hline
Number of BS antennas $M$ &\{$20, 50, 100, 300, 500$\}\\ \hline
Thermal noise density &$-174$~dBm/Hz\\ \hline
Noise Figure  & BS: 9~dB; UE: $5$~dB  \\ \hline
Dynamic range $\kappa^{-1}$  & \{$50, 60, 70,80$\}~dB  \\ \hline
Minimum distance constraints &BS-BS: $40$~m; BS-UE: $10$~m; UE-UE: $3$~m  \\ \hline
Shadowing standard deviation & BS-UE: $10$~dB; BS-BS: $12$~dB; UE-UE: $6$~dB \\ \hline
Pathloss models for BS-UE, UE-UE,& Refer to~\cite{3gppdtdd} \\ 
 and BS-BS channels&\\ \hline
 Fast fading channels & i.i.d.~$\mathcal{CN}(0,1)$\\ \hline
 Propagation loss of self-interference channels&$40$~dB~\cite{ashu13}\\ \hline
\end{tabular}}
  \vspace{-1mm}
\caption{Simulation parameters}
\label{simulation:table}
\end{table}

The average full-duplex spectral efficiency gain is illustrated in Fig.~\ref{simulation:results1} with varying  numbers of BS antennas, where the dynamic range parameter $\kappa$ is $-60$~dB. We verify that the full-duplex gains in realistic network scenarios exist for a range of finite number of BS antennas and the gains will scale with BS antennas. Fig~\ref{results1} depicts the average spectral efficiency gain of the full-duplex system for different dynamic range parameters~$\kappa$ when $M=100$,  which demonstrates the impact of imperfect SI cancellation. Since all users are half-duplex, the downlink gains are not affected by $\kappa$. However, since all the BSs are full-duplex, the uplink gains are severely affected by the dynamic range values especially when the dynamic range value is low. We can see that larger dynamic range $\kappa^{-1}$, (i.e., smaller $\kappa$) will result in less residual SI, thus increasing the uplink gain. Once the dynamic range $\kappa^{-1}$ exceeds certain threshold, there is not much impact of the residual SI on the uplink performance.
\begin{figure*}[t!]
        \centering
        \includegraphics[width=0.5\textwidth,trim = 15mm
      65mm 20mm 65mm,clip]{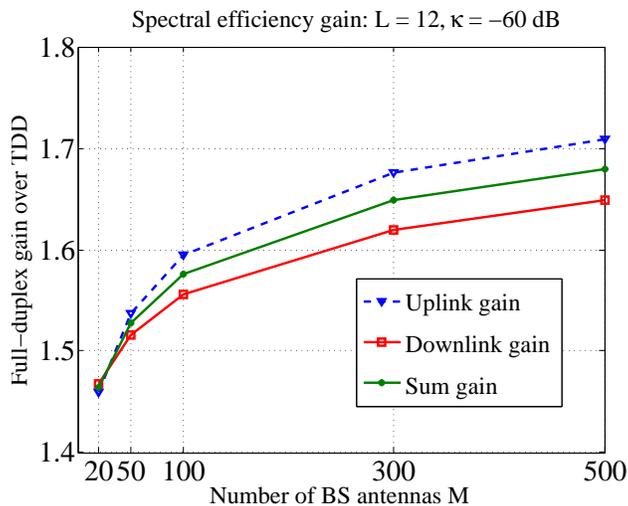} 
    \caption{Average full-duplex spectral efficiency gain with BS antenna arrays when $L=12$ and $\kappa=-60$. } \label{simulation:results1}
\end{figure*}

\begin{figure*}[t!]
        \centering
        \includegraphics[width=0.48\textwidth,trim = 15mm
      65mm 20mm 65mm,clip]{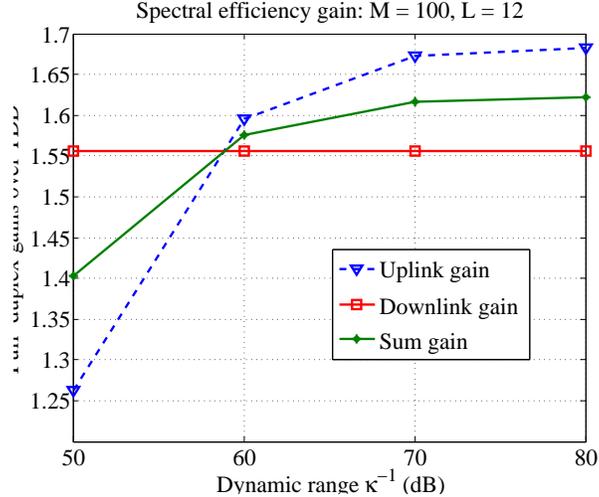} 
    \caption{Average full-duplex spectral efficiency gain with different dynamic ranges when $M=100$ and $L=12$. } \label{results1}
\end{figure*}



\section{Conclusion}\label{conclusion}
In this paper, we have investigated the multi-cell MU-MIMO full-duplex networks where single-antenna full-duplex and half-duplex UEs are served by the full-duplex BSs with multiple antennas. Using low complexity linear receivers and precoders, the ergodic achievable rates of uplink and downlink are characterized for the full-duplex system. Several practical constraints are modeled in the analysis such as imperfect full-duplex radio chains, channel estimation error, pilot overhead and pilot contamination. The large scale system performance is analyzed when each BS has a large antenna array. When the number of BS antennas grows infinitely large, the 2$\times$ asymptotic full-duplex spectral efficiency gain can be achieved over the TDD system with perfect CSI. When channel estimation error and channel training overhead are considered, the 2$\times$ asymptotic full-duplex spectral efficiency gain is achieved when serving only full-duplex UEs. Furthermore, we show by numerical simulation that full-duplex system can achieve both spectral efficiency gain and antenna reduction as the full-duplex system can require one order of magnitude fewer antennas than TDD to achieve the same or better performance in certain cases.

\appendix
\subsection{Proof of Proposition~\ref{prop1}}\label{proof1}
Based on Jensen's inequality and the convexity of $\mathrm{log}_2\left(1+x^{-1}\right)$, we have $\mathbb{E}\left[\mathrm{log}_2\left(1+x^{-1}\right)\right]\geq \mathrm{log}_2\left(1+\mathbb{E}^{-1}(x)\right)$. Hence we can lower bound the ergodic achievable rate of the $n$-th uplink UE in the $j$-th cell in~(\ref{ul:lb1}) as
\begin{gather}
\begin{aligned}
&\tilde{R}^{fd,p}_{u,jn}\geq R^{fd,p}_{u,jn}\triangleq \mathrm{log}_2\Bigg(1+\\
&\mathbb{E}^{-1}\left[\frac{P_u\sum_{(l,m)\neq (j,n)}|\bm{g}^H_{u,jjn}\bm{g}_{u,jlm}|^2+I_{bs-bs}+\norm{\bm{g}_{u,jjn}}^2(\sigma^2+\kappa P_d\beta_{b,jj})}{P_u\norm{\bm{g}_{u,jjn}}^4}\right]\Bigg)\\
&=\mathrm{log}_2\left(1+\mathbb{E}^{-1}\left[\frac{P_u\sum_{(l,m)\neq (j,n)}|\tilde{g}_{u,jlm}|^2+\sum_{l\neq j}\sum_{k\in\mathcal{K}_d}\frac{P_d}{K_d \gamma_l}|\tilde{v}_{jlk}|^2+\sigma^2+\kappa P_d\beta_{b,jj}}{P_u\norm{\bm{g}_{u,jjn}}^2}\right]\right), \label{ul:lb2}
\end{aligned}
\end{gather}
where $\tilde{g}_{u,jlm}\triangleq \frac{\bm{g}^H_{u,jjn}\bm{g}_{u,jlm}}{\norm{\bm{g}_{u,jjn}}}$, $\tilde{v}_{jlk}\triangleq \frac{\bm{g}^H_{u,jjn}\bm{V}_{jl}\bm{g}^*_{d,llk}}{\norm{\bm{g}_{u,jjn}}}$. Conditioned on $\bm{g}_{u,jjn}$, $\tilde{g}_{u,jlm}\sim\mathcal{CN}(0,\beta_{u,jlm})$ and $\tilde{v}_{jlk}\sim\mathcal{CN}(0,M\beta_{b,jl}\beta_{d,llk})$ are independent of $\bm{g}_{u,jjn}$. Thus we have
\begin{gather}
\begin{aligned}
&\mathbb{E}\left[\frac{P_u\sum_{(l,m)\neq (j,n)}|\tilde{g}_{u,jlm}|^2+\sum_{l\neq j}\sum_{k\in\mathcal{K}_d}\frac{P_d}{K_d \gamma_l}|\tilde{v}_{jlk}|^2+\sigma^2+\kappa P_d\beta_{b,jj}}{P_u\norm{\bm{g}_{u,jjn}}^2}\right]\\
&=\left(\sum_{(l,m)\neq (j,n)}P_u\mathbb{E}[|\tilde{g}_{u,jlm}|^2]+\sum_{l\neq j}\sum_{k\in\mathcal{K}_d}\frac{P_d}{K_d \gamma_l}\mathbb{E}[|\tilde{v}_{jlk}|^2]+\sigma^2+\kappa P_d\beta_{b,jj}\right)\mathbb{E}\left[\frac{1}{P_u\norm{\bm{g}_{u,jjn}}^2}\right]. \label{ul:lb3}
\end{aligned}
\end{gather}

From Lemma~2.10 in \cite{tulino2004random}, for a central Wishart matrix $\bm{W}\sim\mathcal{W}_m(n,\bm{I})$ with $n\geq m$, we have 
$\mathbb{E}[\mathsf{tr}\{\bm{W}^{-1}\}]=\frac{m}{n-m}$ for $n>m$. Hence 
\begin{gather}
\begin{aligned}
\mathbb{E}\left[\frac{1}{P_u\norm{\bm{g}_{u,jjn}}^2}\right]=\frac{1}{(M-1)P_u\beta_{u,jjn}}~\text{for}~M\geq 2. \label{ul:lb4}
\end{aligned}
\end{gather}
Combing (\ref{ul:lb2}), (\ref{ul:lb3}) and (\ref{ul:lb4}), we can obtain the uplink achievable rate in (\ref{eq:prop1}).

To derive the lower bound on the downlink achievable rate, we need to evoke the Lemma~2.10 in \cite{tulino2004random} again where $\mathbb{E}[\mathsf{tr}\{\bm{W}^{-2}\}]=\frac{mn}{(n-m)^3-(n-m)}$ for $n>m+1$. Thus we have 
\begin{gather}
\begin{aligned}
\mathbb{E}\left[\frac{1}{\norm{\bm{g}_{d,llk}}^4}\right]=\frac{1}{(M-1)(M-2)\beta_{d,llk}^2}~\text{for}~M\geq 3. 
\end{aligned}
\end{gather}
Following the same approach used in the derivation of uplink rate, we can obtain the achievable rate of each downlink user given in proposition~\ref{prop1}. The details are omitted to avoid redundancy.

\subsection{Proof of Proposition~\ref{prop2}}\label{proof2}
Similar to the proof of Proposition~\ref{prop1}, applying Jensen's inequality, we can lower bound the ergodic achievable rate of the $n$-th uplink user in the $j$-th cell in~(\ref{ulim:lb1}) as
\begin{gather}
\begin{aligned}
&\tilde{R}^{fd,ip}_{u,jn}\geq R^{fd,ip}_{u,jn}\triangleq \mathrm{log}_2\Bigg(1+\\
&\mathbb{E}^{-1}\left[\frac{P_u\abs{\tilde{\epsilon}_{jjn}}^2+P_u\sum_{(l,m)\neq (j,n)}|\tilde{\tilde{g}}_{u,jlm}|^2+\sum_{l\neq j}\sum_{k\in\mathcal{K}_d}\frac{P_d}{K_d \gamma_l}|\tilde{\tilde{v}}_{jlk}|^2+\sigma^2+\kappa P_d\beta_{b,jj}}{P_u\norm{\hat{\bm{g}}_{u,jjn}}^2}\right]\Bigg)\label{ulim:lb2}
\end{aligned}
\end{gather}
where $\tilde{\epsilon}_{jjn}\triangleq \frac{\hat{\bm{g}}^H_{u,jjn}\bm{\epsilon}_{u,jjn}}{\norm{\hat{\bm{g}}_{u,jjn}}}$, $\tilde{\tilde{g}}_{u,jlm}\triangleq \frac{\hat{\bm{g}}^H_{u,jjn}\bm{g}_{u,jlm}}{\norm{\hat{\bm{g}}_{u,jjn}}}$, $\tilde{\tilde{v}}_{jlk}\triangleq \frac{\hat{\bm{g}}^H_{u,jjn}\bm{V}_{jl}\hat{\bm{g}}^*_{d,llk}}{\norm{\bm{g}_{u,jjn}}}$. Conditioned on $\hat{\bm{g}}_{u,jjn}$, $\tilde{\epsilon}_{jjn}\sim\mathcal{CN}\left(0,\frac{\beta_{u,jjk}\left(\sigma^2+P_{tr}\sum_{l\neq j}\beta_{u,jlk}\right)}{\lambda_{u,jk}}\right)$, $\tilde{\tilde{v}}_{jlk}\sim\mathcal{CN}\left(0,M\beta_{b,jl}\frac{P_{tr}\beta^2_{d,llk}}{\lambda_{d,lk}}\right)$ and $\tilde{\tilde{g}}_{u,jlm}\sim(0,\mathrm{var}(\tilde{\tilde{g}}_{u,jlm}))$ are independent of $\hat{\bm{g}}_{u,jjn}$.
Using (\ref{eq:mmse}), we have
\begin{gather}
\begin{aligned}
\tilde{\tilde{g}}_{u,jlm}=\frac{P_{tr}\beta_{u,jjn}}{\lambda_{u,jn}}\left(\sum_{l_1=1}^L \bm{g}^H_{u,jl_1n}+\frac{\bm{n}^H_{jn}}{\sqrt{P_{tr}}}\right)\frac{\bm{g}_{u,jlm}}{\norm{\hat{\bm{g}}_{u,jjn}}}.
\end{aligned}
\end{gather}
Evoking Lemma~2.9 in \cite{tulino2004random}, where for a central Wishart matrix $\bm{W}\sim\mathcal{W}_m(n,\bm{I})$ with $n\geq m$, $\mathbb{E}[\mathsf{tr}\{\bm{W}^{2}\}]=mn(m+n)$, we can calculate the variance of $\tilde{\tilde{g}}_{u,jlm}$ as
\begin{gather}
\begin{aligned}
\mathbb{E}[\abs{\tilde{\tilde{g}}_{u,jlm}}^2]= 
  \begin{cases}
  \frac{P_{tr}}{\lambda_{u,jn}}\left[(M+1)\beta^2_{u,jln}+\sum_{l_1\neq l}\beta_{u,jl_1 n}\beta_{u,jln}+\frac{\beta_{u,jln}\sigma^2}{P_{tr}}\right] &  ~~m=n,l\neq j\\
  \beta_{u,jlm} &~~ m\neq n.
   \end{cases}
\label{ulim:lb3}
\end{aligned}
\end{gather}
Now we can rewrite the expectation in (\ref{ulim:lb2}) as
\begin{gather}
\begin{aligned}
&\mathbb{E}\left[\frac{P_u\abs{\tilde{\epsilon}_{jjn}}^2+P_u\sum_{(l,m)\neq (j,n)}|\tilde{\tilde{g}}_{u,jlm}|^2+\sum_{l\neq j}\sum_{k\in\mathcal{K}_d}\frac{P_d}{K_d \gamma_l}|\tilde{\tilde{v}}_{jlk}|^2+\sigma^2+\kappa P_d\beta_{b,jj}}{P_u\norm{\hat{\bm{g}}_{u,jjn}}^2}\right]\\
&=\left(P_u\mathbb{E}[|\tilde{\epsilon}_{jjn}|^2]+P_u\sum_{\mathclap{(l,m)\neq (j,n)}}\mathbb{E}[|\tilde{\tilde{g}}_{u,jlm}|^2]+\sum_{l\neq j}\sum_{k\in\mathcal{K}_d}\frac{P_d}{K_d \gamma_l}\mathbb{E}[|\tilde{\tilde{v}}_{jlk}|^2]+\sigma^2+\kappa P_d\beta_{b,jj}\right)\mathbb{E}\left[\frac{1}{P_u\norm{\hat{\bm{g}}_{u,jjn}}^2}\right]. \label{ulim:lb4}
\end{aligned}
\end{gather}
Combing (\ref{ulim:lb2}), (\ref{ulim:lb3}) and (\ref{ulim:lb4}), we can obtain the uplink achievable rate in (\ref{eq:prop2}).

For the downlink achievable rate in (\ref{dlim:lb1}), we first compute $\mathbb{E}\left[\bm{g}^T_{d,llk}\hat{\bm{g}}^*_{d,llk}\right]$. Let $\mu=\bm{g}^T_{d,llk}\hat{\bm{g}}^*_{d,llk}$, since $\bm{g}_{d,llk}=\hat{\bm{g}}_{d,llk}+\bm{\epsilon}_{d,llk}$ and $\hat{\boldsymbol{g}}_{d,llk}$ is independently of $\boldsymbol{\epsilon}_{d,llk}$, we have 
\begin{gather}
\begin{aligned}
\mathbb{E}[\mu]&=\mathbb{E}\left[(\hat{\bm{g}}^T_{d,llk}+\bm{\epsilon}^T_{d,llk})\hat{\bm{g}}^*_{d,llk}\right]\\
&=\mathbb{E}\left[\norm{\hat{\bm{g}}_{d,llk}}^2\right]=\frac{MP_{tr}\beta^2_{d,llk}}{\lambda_{d,lk}}.
\end{aligned}
\end{gather}
Again invoking  Lemma~2.9 in \cite{tulino2004random}, we have
\begin{gather}
\begin{aligned}
\mathbb{E}\left[\mu^2\right]&=\mathbb{E}\left[\norm{\hat{\bm{g}}_{d,llk}}^4\right]+\mathbb{E}\left[\bm{\epsilon}^T_{d,llk}\hat{\bm{g}}^*_{d,llk}\hat{\bm{g}}^T_{d,llk}\bm{\epsilon}^*_{d,llk}\right]\\
&=\frac{M(M+1)P^2_{tr}\beta^4_{d,llk}+MP_{tr}\beta^3_{d,llk}(\sigma^2+P_{tr}\sum_{j\neq l}\beta_{d,ljk})}{\lambda_{d,lk}^2}.
\end{aligned}
\end{gather}

Since $\mathsf{var}\left[\bm{g}^T_{d,llk}\hat{\bm{g}}^*_{d,llk}\right]\triangleq\mathsf{var}(\mu)=\mathbb{E}(\mu^2)-\mathbb{E}^2(\mu)$, we can obtain that
\begin{gather}
\begin{aligned}
\mathsf{var}(\mu)&=\frac{MP_{tr}\beta^3_{d,llk}}{\lambda_{d,lk}}.
\end{aligned}
\end{gather}

Next, we compute $\mathbb{E}\left[\abs{\bm{g}^T_{d,jlk}\hat{\bm{g}}^*_{d,jji}}^2\right]$ as
\begin{gather}
\begin{aligned}
\mathbb{E}\left[\abs{\bm{g}^T_{d,jlk}\hat{\bm{g}}^*_{d,jji}}^2\right]= 
  \begin{cases}
  \frac{MP^2_{tr}\beta^2_{d,jjk}}{\lambda^2_{d,jk}}\!\!\left[(M+1)\beta^2_{d,jlk}+\sum_{l_1\neq l}\beta_{d,jl_1 k}\beta_{d,jlk}+\frac{\beta_{d,jlk}\sigma^2}{P_{tr}}\right] &  ~~i=k,j\neq l\\
   \frac{MP_{tr}\beta^2_{d,jji}\beta_{d,jlk}}{\lambda_{d,ji}} &~~ i\neq k.
   \end{cases}
\end{aligned}
\end{gather}
The rest terms in (\ref{dlim:lb1}) can be calculated easily, and thus the details are omitted. Combing all the results above, we can obtain downlink achievable rate given in proposition~\ref{prop2}.

\bibliographystyle{IEEEtran}
\scriptsize
\bibliography{Reference}

\end{document}